\documentclass{sig-alternate-10pt}
\usepackage{color}
\usepackage{multirow}
\usepackage{url}
\usepackage{algorithm}
\usepackage[noend]{algpseudocode}
\usepackage{subfigure}
\usepackage{enumitem} 
\usepackage{times}
\usepackage{xspace}
\usepackage[table]{xcolor}
\usepackage{latexsym}
\usepackage{amsmath}
\usepackage{tabularx}
\usepackage{balance}
\usepackage[hidelinks]{hyperref}
\usepackage[square,sort,comma,numbers]{natbib}

\usepackage{tlatex}

\usepackage{varwidth}

\setlist{nolistsep}
\newif\ifload
\loadfalse

\setlist{nolistsep}
\newcommand{\para}[1]{\medskip\noindent\textbf{#1}}

\newcommand{\paraf}[1]{\noindent\textbf{#1}}
\newcommand{\cut}[1]{}

\newcommand{\sysname}{Harmonia\xspace}
\newcommand{\seq}{\textit{seq}\xspace}

\newcommand{\msg}[1]{\textsc{#1}}

\newtheorem{theorem}{Theorem}

\begin{document}
\sloppy
\date{}

\title{\sysname: Near-Linear Scalability for Replicated Storage \\
with In-Network Conflict Detection}

\author{
Hang Zhu$^1$, Zhihao Bai$^1$, Jialin Li$^2$, Ellis Michael$^2$,
Dan Ports$^3$, Ion Stoica$^4$, Xin Jin$^1$ \\
\\
\affaddr{$^1$Johns Hopkins University, $^2$University of Washington, $^3$Microsoft Research, $^4$UC Berkeley}
}

\maketitle

\begin{abstract}
Distributed storage employs replication to mask failures and improve
availability. However, these systems typically exhibit a hard tradeoff between
consistency and performance. Ensuring consistency introduces coordination
overhead, and as a result the system throughput does not scale with the number of
replicas. We present \sysname, a replicated storage architecture that exploits
the capability of new-generation programmable switches to obviate this
tradeoff by providing near-linear scalability without sacrificing consistency.
To achieve this goal, \sysname detects read-write conflicts in the network,
which enables any replica to serve reads for objects with no pending writes.
\sysname implements this functionality at line rate, thus imposing no
performance overhead. We have implemented a prototype of \sysname on a cluster
of commodity servers connected by a Barefoot Tofino switch, and have integrated
it with Redis. We demonstrate the generality of our approach by
supporting a variety of replication protocols, including primary-backup,
chain replication, Viewstamped Replication, and NOPaxos. Experimental results
show that \sysname improves the throughput of these protocols by up to 10$\times$ for a
replication factor of 10, providing near-linear scalability up to the
limit of our testbed.
\end{abstract}

\newcommand{\secname}{Introduction}
\section{\secname}
\label{sec:introduction}

Replication is one of the fundamental tools in the modern distributed
storage developer's arsenal. Failures are a regular appearance in
large-scale distributed systems, and strongly consistent replication can
transparently mask these faults to achieve high system availability. However,
it comes with a high performance cost.

One might hope that adding more servers to a replicated system would
increase not just its reliability but also its system performance---ideally,
providing \emph{linear scalability} with the number of replicas. The
reality is quite the opposite: performance decreases with more
replicas, as an expensive replication protocol needs to be run to ensure that all
replicas are consistent. Despite much effort to reduce the cost of
these protocols, the best case is a system that approaches the
performance of a single node~\cite{chain-replication,nopaxos}.

Can we build a strongly consistent replicated system that approaches
linear scalability? Write operations inherently need to be applied to
all replicas, so more replicas cannot increase the write throughput. However,
many real-world workloads are highly skewed towards reads~\cite{gfs,
  memcache-nsdi13}---with read:write ratios as high as
30:1~\cite{workload-fb-sigmetrics12}. A scalable but naive approach is to
allow \emph{any individual replica} to serve a read, permitting
the system to achieve near-linear scalability for such workloads. Yet
this runs afoul of a fundamental limitation. Individual replicas may
lag behind, or run ahead of, the consensus state of the group. Thus,
serving reads from any storage replica has the potential to return
stale or even uncommitted data, compromising the consistency guarantee
of the replicated system. Protocol-level solutions like CRAQ~\cite{craq}
require extra server coordinations, and thus
inevitably impose additional performance overheads.

We show that it is possible to circumvent this limitation and
simultaneously achieve near-linear scalability and consistency for
replicated storage.
We do so with \sysname, a new replicated storage architecture that exploits the
capability of new-generation programmable switches. Our key observation is that while
individual replicas may constantly diverge from the consensus state, the set of
\emph{inconsistent data} at any given time is small. A storage system may store
millions or billions of objects or files. Of these, only the ones
that have writes in progress---i.e., the set of objects actively being
updated---may be inconsistent. For the others, any replica can safely serve a
read. Two features of many storage systems make this an especially
powerful observation: $(i)$ read-intensive workloads in real-world
deployments~\cite{memcache-nsdi13, workload-fb-sigmetrics12}
mean that fewer objects are written over time, and $(ii)$ emerging in-memory
storage systems~\cite{ramcloud, redis, memcached,
facebook-tao-sigmod12} complete writes faster, reducing the interval
of inconsistency.

The challenge in leveraging this insight lies in efficiently detecting which
objects are dirty, i.e., have pending updates. Implementing this functionality
in a server would make the system be bottlenecked by the server, instead of scaling
out with the number of storage replicas.
\sysname demonstrates that this
idea can be realized on-path in the network with programmable switches at line
rate, with no performance penalties. The core component of \sysname is a
read-write conflict detection module in the switch data plane that monitors all
requests to the system and tracks the dirty set. The switch detects whether a
read conflicts with pending writes, and if not, the switch sends it
directly to one of the replicas. Other requests are executed according to the
normal protocol. This design exploits two principal benefits of
in-network processing: $(i)$ the central location of the switch on the data path
that allows it to monitor traffic to the cluster, and $(ii)$ its capability for
line-rate, near-zero overhead processing.

\sysname can be viewed as a new take on network anycast in the context of
replicated storage. Different from recent work that directly embeds \emph{application
data and logic} into programmable switches~\cite{netcache, netchain, netpaxos,
netpaxos-ccr} , \sysname still uses switches for \emph{forwarding}, but in an
application-aware manner by tracking \emph{application metadata}
(contended object IDs, not values) with much less switch memory.

\sysname is a general approach. It augments existing replication protocols
with the ability to serve reads from any replica, and does not sacrifice
fault tolerance or strong consistency (i.e., linearizability).
As a result, it can be applied to both major classes of replication
protocols---primary-backup and quorum-based. We have applied
\sysname to a range of representative protocols, including
primary-backup~\cite{primary-backup}, chain
replication~\cite{chain-replication},
Viewstamped Replication~\cite{viewstamped,liskov12:_views_replic_revis}, and NOPaxos~\cite{nopaxos}.

In summary, this paper demonstrates that:
\begin{itemize}[leftmargin=*]
  \item The \sysname architecture can achieve near-linear scalability with
  near-zero overhead by moving conflict detection to an in-network component.
  (\S\ref{sec:overview}, \S\ref{sec:design})

  \item The \sysname conflict detection mechanism can be implemented in
  the data plane of new-generation programmable switches and run at
  line rate. (\S\ref{sec:switch})

  \item Many replication protocols can be integrated with \sysname while
  maintaining linearizability. (\S\ref{sec:protocols})
\end{itemize}

\medskip
We implement a \sysname prototype using a cluster of servers
connected by a Barefoot Tofino switch and integrate it with Redis.
Our experiments show that \sysname
improves the throughput of the replication
protocols by up to 10$\times$ for a replication factor of 10, providing near-linear
scalability up to the limit of our testbed.
We provide a proof of correctness and a model-checked TLA+ specification as appendices.
Of course, the
highest possible write throughput is that of \emph{one} replica, since
writes have to be processed by all replicas. This can be achieved by
chain replication~\cite{chain-replication} and NOPaxos~\cite{nopaxos}.
\sysname fills in the missing piece for reads: it demonstrates how to make reads
scalable without sacrificing either write performance or consistency.


\renewcommand{\secname}{Background}
\section{\secname}
\label{sec:background}

An ideal replicated system provides single-system
linearizability~\cite{HerlihyWing90}---it appears as though
operations are being executed, one at a time, on a single replica, in
a way that corresponds to the timing of operation invocations and
results. Many replication protocols can be used to ensure
this property. They fall primarily into two
classes---primary-backup protocols and quorum-based protocols.

\para{Primary-backup.} The primary-backup protocol~\cite{primary-backup}
organizes a system into a \emph{primary} replica, which is responsible for
determining the order and results of operations, and a set of \emph{backup}
replicas that execute operations as directed by the primary. This is typically
achieved by having the primary transfer state updates to the replicas after
operation execution. At any time, only one primary is in operation. Should it
fail, one of the backup replicas is promoted to be the new primary---a task
often mediated by an external configuration service~\cite{zookeeper,chubby} or
manual intervention. The primary-backup protocol is able to tolerate $f$ node
failures with $f$+$1$ nodes.

The primary-backup protocol has many variants. Chain
replication~\cite{chain-replication} is a high-throughput variant, and is
used in many storage systems~\cite{fawn, flexkv, hyperdex}. It organizes the replicas
into a chain. Writes are sent to the head and propagated to the tail;
reads are directly processed by the tail. The system throughput is
bounded by a single server---the tail.

\para{Quorum-based protocols.} Quorum-based protocols such as
Paxos~\cite{paxos98,paxos01} operate by ensuring that each operation is executed
at a quorum---typically a majority---of replicas before it is considered
successful. While they seem quite different from primary-backup protocols, the
conceptual gap is not as wide as it appears in practice. Many Paxos deployments
use the Multi-Paxos optimization~\cite{paxos01} (or, equivalently, Viewstamped
Replication~\cite{viewstamped} and Raft~\cite{raft}). One of the replicas
runs the first phase of the protocol to elect itself as a stable \emph{leader}
until it fails. It can then run the second phase repeatedly to execute multiple
operations and commit to other replicas, which is very similar to the
primary-backup protocol. System throughput is largely determined by the number
of messages that needs to be processed by the bottleneck node, i.e., the leader. A
common optimization allows the leader to execute reads without coordinating with
the others, by giving the leader a lease. Ultimately, however,
the system throughput is limited to
that of one server.

\renewcommand{\secname}{Towards Linear Scalability}
\section{\secname}
\label{sec:approach}

The replication protocols above can achieve, at best, the throughput
of a single server. With judicious optimization, they can allow reads
to be processed by one designated replica---the tail in chain
replication or the leader in Multi-Paxos. That single replica then
becomes the bottleneck. Read scalability, i.e., making system
throughput scale with the number of replicas, requires going further.

Could we achieve read scalability by allowing reads to be processed by
\emph{any} replica, not just a single designated one, without
coordination? Unfortunately, naively doing so could violate
consistency. Updates cannot be applied instantly across all the
replicas, so at any given moment some of the replicas may not be
consistent. We categorize the resulting anomalies into two kinds.

\para{Read-ahead anomalies.} A read-ahead anomaly occurs when a
replica returns a result that has not yet been committed. This might
reflect a future state of the system, or show updates that will never
complete. Neither is correct.

Consider the case of chain replication, where each replica would
answer a read with its latest known state. Specifically, suppose
there are three nodes, and the latest update to an object has been
propagated to nodes 1 and 2. A read on this object sent to either of
these nodes would return the new value. While this may not necessarily
seem problematic, it is a violation of linearizability. A subsequent
request to node 3 could still return the old value---causing a client
to see an update appearing and disappearing depending on which replica
it contacts.

\para{Read-behind anomalies.}
One might hope that these anomalies could be avoided by requiring
replicas to return the latest \emph{known committed} value.
Unfortunately, this introduces a second class of anomalies, where some
replicas may return a stale result that does not reflect the latest
updates.

Consider a Multi-Paxos deployment, in which replicas only execute an
operation once they have been notified by the leader that consensus
has been reached for that operation. Suppose that a client submits a
write to an object, and consider the instant just after the leader receives
the last response in a quorum. It can then execute the operation and
respond to the client. However, the other replicas do not know that
the operation is successful. If the client then executes a read to one
of the other replicas, and it responds---unilaterally---with its
latest committed value, the client will not see the effect of its
latest write.

\para{Protocols.}
We classify replication protocols based on the anomalies.
We refer to protocols that have each type of
anomalies as \emph{read-ahead protocols} and \emph{read-behind
protocols}, respectively. Of the protocols we discuss in this paper,
primary-backup and chain replication are read-ahead protocols, and
Viewstamped Replication/Multi-Paxos and NOPaxos are read-behind
protocols. Note that although the primary-backup systems are
read-ahead and the quorum systems are read-behind, this is not
necessarily true in general; read-ahead quorum protocols are also
possible, for example.

\subsection{\sysname Approach}

How, then, can we \emph{safely} and \emph{efficiently} achieve read scalability, without
sacrificing linearizability? The key is to view the system at the
\emph{individual object level}. The majority of objects are
quiescent, i.e., have no modifications in progress. These objects will be consistent
across at least a majority of replicas. In that case, any
replica can unilaterally answer a read for the object. While
modifications to an object \emph{are} in progress, reads
on the object must follow the full protocol.

Conceptually, \sysname achieves read scalability by introducing a new
component to the system, a request scheduler.
The request scheduler
monitors requests to the system to detect conflicts
between read and write operations.
Abstractly, it maintains a table of
objects in the system and tracks whether they are contended or uncontended,
i.e., the \emph{dirty set}. When there is no conflict, it directs reads to
any replica. The request is flagged so that the replica can
respond directly. When conflicts are detected, i.e., a concurrent
write is in progress, reads follow the normal protocol.

To allow the request scheduler to detect conflicts, it needs to be able to
interpose on all read and write traffic to the system. This means that the
request scheduler must be able to achieve very high throughput---implementing
the conflict detection in a server would still make the entire system be
bottlenecked by the server. Instead, we implement the request scheduler in the
network itself, leveraging the capability of programmable switches to run at
line rate, imposing no performance penalties.

\begin{figure}[t]
\centering
    \includegraphics[width=0.95\linewidth]{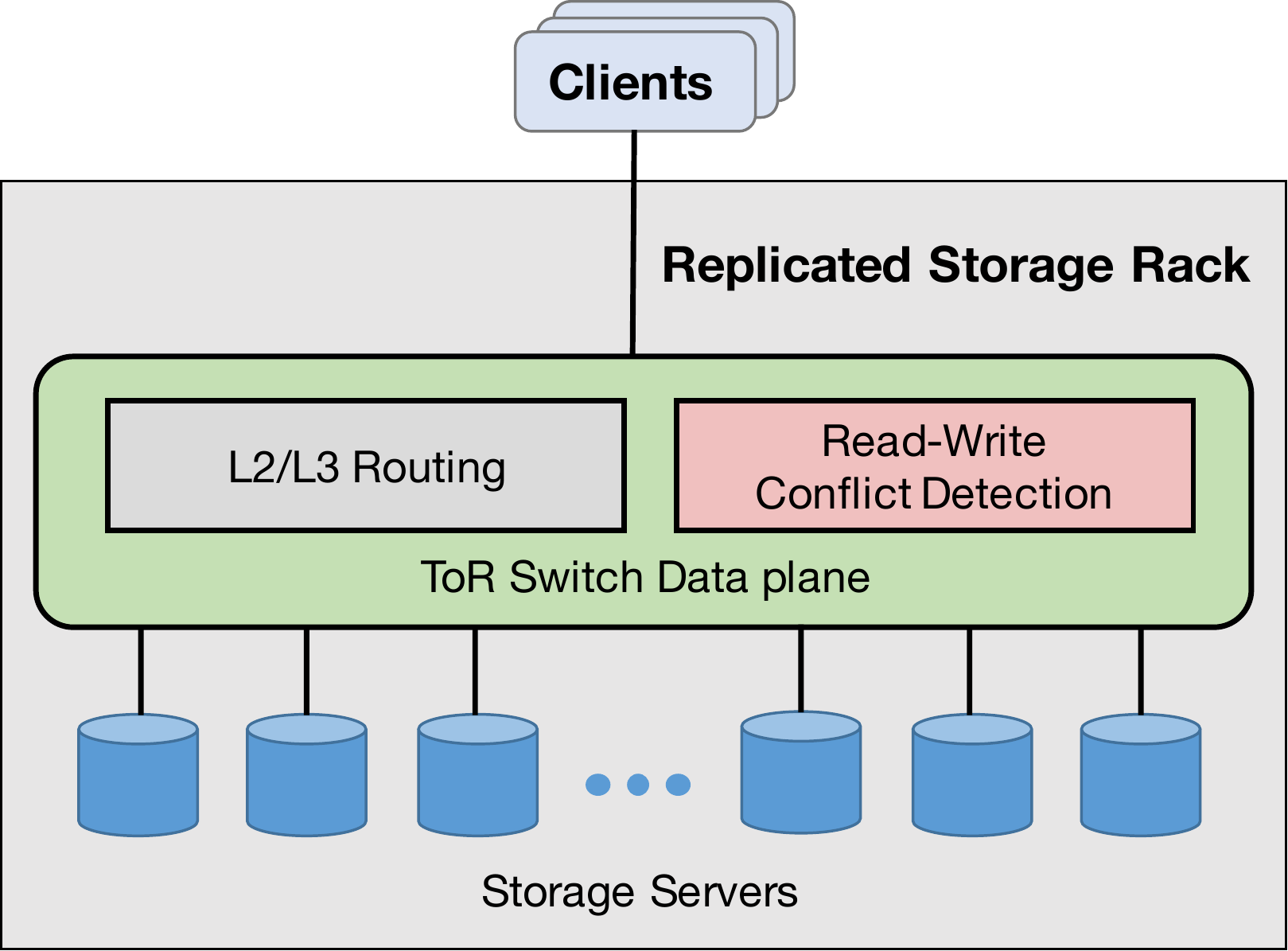}
\caption{\sysname architecture.}
\label{fig:overview_architecture}
\end{figure}

Conflict detection has been used before to achieve read scalability
for certain replicated systems. Specifically, CRAQ~\cite{craq}
provides read scalability for chain replication by tracking contended
and uncontended objects at the protocol level. This requires
introducing an extra phase to notify replicas when an object is clean
vs. dirty. \sysname's in-switch conflict detection architecture
provides two main benefits. First, it generalizes the approach to
support many different protocols---as examples, we have used \sysname
with primary-backup, chain replication, Viewstamped Replication, and
NOPaxos. Supporting the diverse range of replication protocols are in
use today is important because they occupy different points in the
design space: latency-optimized vs. throughput-optimized,
read-optimized vs. write-optimized, storage overhead vs. performance
under failure, etc. CRAQ is specific to chain replication, and it is
not clear that it is possible to extend its protocol-level approach to
other protocols. Second, \sysname's in-switch implementation avoids
imposing additional overhead to track the dirty set. As we show in
Section~\ref{sec:evaluation:generality}, CRAQ is able to achieve read
scalability only at the expense of a decrease in write throughput.
\sysname has no such cost.

\subsection{Challenges}

Translating the basic model of the request scheduler above to a working
implementation presents several challenges:

\begin{enumerate}
\item How can we expose system state to the request scheduler so that it can
  implement conflict detection?
\item How do we ensure the switch's view of which
  objects are contended matches the system's reality, even as messages
  are dropped or reordered by the network? Errors here may
  cause clients to see invalid results.
\item How do we implement this functionality fully within a switch
  data plane? This drastically limits computational and storage capacity.
\item What modifications are needed to replication protocols to ensure
  they provide linearizability when integrated with \sysname?
\end{enumerate}


\renewcommand{\secname}{Harmonia Architecture}
\section{\secname}
\label{sec:overview}

\sysname is a new replicated storage architecture that achieves
near-linear scalability without sacrificing consistency using
in-network conflict detection. This is implemented using an in-switch
request scheduler, which is located on the path between the
clients and server nodes. In many enterprise and cloud scenarios where storage
servers are located in a dedicated storage rack, this
can be conveniently achieved by using the top-of-rack (ToR) switch,
as shown in \autoref{fig:overview_architecture}.
We discuss other, more scalable deployment options in \S\ref{sec:switch:issue}.

\para{Switch.} The switch implements the \sysname request scheduler.
It is responsible for detecting read-write conflicts.
It behaves as a
standard L2/L3 switch, but provides additional conflict detection
functionality for packets with a reserved L4 port. This makes \sysname fully
compatible with normal network functionality.

The \emph{read-write conflict detection module} identifies whether a read
has a conflict with a pending write. It does this by maintaining a
sequence number, a dirty set and the last-committed point (\S\ref{sec:design}). We show
how to handle requests with this module while \emph{guaranteeing consistency}
(\S\ref{sec:design}), and how to use the register arrays to design a
hash table supporting the necessary operations \emph{at line rate}
(\S\ref{sec:switch}).

While the \sysname switch can be rebooted or replaced and is not a
single point of failure of the storage system, there is only a single active \sysname switch
for conflict detection at any time. The replication protocol enforces this invariant by
periodically agreeing on which switch to use for each time slice
(\S\ref{sec:design:failure}).

\para{Storage servers.} The storage servers store objects and serve
requests, using a replication protocol for consistency and fault
tolerance. \sysname requires minimal modifications to the replication
protocol (\S\ref{sec:protocols}). It incorporates
a shim layer in each server to translate custom \sysname request packets to API
calls to the storage system.

\para{Clients.} \sysname provides a client library for applications to
access the storage system, which provides a similar
interface as existing storage systems.
e.g., \texttt{GET} and
\texttt{SET} for Redis~\cite{redis} in our prototype.
The library translates between API calls and \sysname packets. It exposes
two important fields in the packet header to \sysname switch: the operation type
(read or write), and the affected object ID.


\renewcommand{\secname}{In-Network Conflict Detection}
\section{\secname}
\label{sec:design}


\paraf{Key idea.} \sysname employs a switch as a conflict detector,
which tracks the dirty set, i.e., the set of contended objects.
While the available switch memory is limited, the set of objects with
outstanding writes is small compared to the overall storage size of
the system, making this set readily implementable on the switch.

To implement conflict detection,
a \sysname switch tracks three
pieces of state:
$(i)$ a \emph{monotonically-increasing sequence number},\footnote{We use the term
sequence number here for simplicity. Sequentiality
is not necessary; a strictly increasing timestamp would suffice.}
which is incremented and inserted into each write,
$(ii)$ a \emph{dirty set}, which additionally tracks the largest sequence number of
the outstanding writes to each object, and
$(iii)$ the \emph{last-committed point}, which tracks the sequence number of the
latest write committed by the system known to the switch.

\begin{algorithm}[t!]
\caption{ProcessRequestSwitch(pkt)}
\begin{itemize}[noitemsep,nolistsep]
    \item[--] $seq$: sequence number at switch
    \item[--] $dirty\_set$: map containing largest sequence number for each object with pending writes
    \item[--] $last\_committed$: largest known committed sequence number
\end{itemize}
\begin{algorithmic}[1]
\If{$pkt.op == \msg{write}$}
    \State $\seq \gets \seq + 1$
    \State $pkt.seq \gets seq$
    \State $dirty\_set.put(pkt.obj\_id, seq)$
\ElsIf{$pkt.op == \msg{write-completion}$}
    \If{$pkt.seq \ge dirty\_set.get(pkt.obj\_id)$}
        \State $dirty\_set.delete(obj\_id)$
    \EndIf
    \State
    \begin{varwidth}[t]{\linewidth}
    $last\_committed \gets max(last\_committed, pkt.seq)$
    \end{varwidth}
\ElsIf{$pkt.op == \msg{read}$}
    \If{$ \neg dirty\_set.contains(pkt.\mathit{obj\_id})$}
        \State $pkt.last\_committed \gets last\_committed$
        \State $pkt.dst \gets$ random replica
    \EndIf
\EndIf
\State Update packet header and forward
\end{algorithmic}
\label{alg:switch}
\end{algorithm}

The dirty set allows the switch to detect when a read contends with
ongoing writes. When they do not, \sysname can send the read to a
single random replica for better performance. Otherwise, these reads
are passed unmodified to the underlying replication protocol. The sequence number
disambiguates concurrent writes to the same object, permitting the
switch to store \emph{only one} entry per contended object in the
dirty set. The last-committed sequence number is used to ensure
linearizability in the face of reordered messages, as will be described
in \S\ref{sec:design:reordering}.



We now describe in detail the interface and how it is used.
We use the primary-backup
protocol as an example in this section, and describe adapting other
protocols in \S\ref{sec:protocols}.

\begin{figure*}[t]
    \centering
    \subfigure[Write.]{
        \label{fig:design_query1}
        \includegraphics[width=0.45\linewidth]{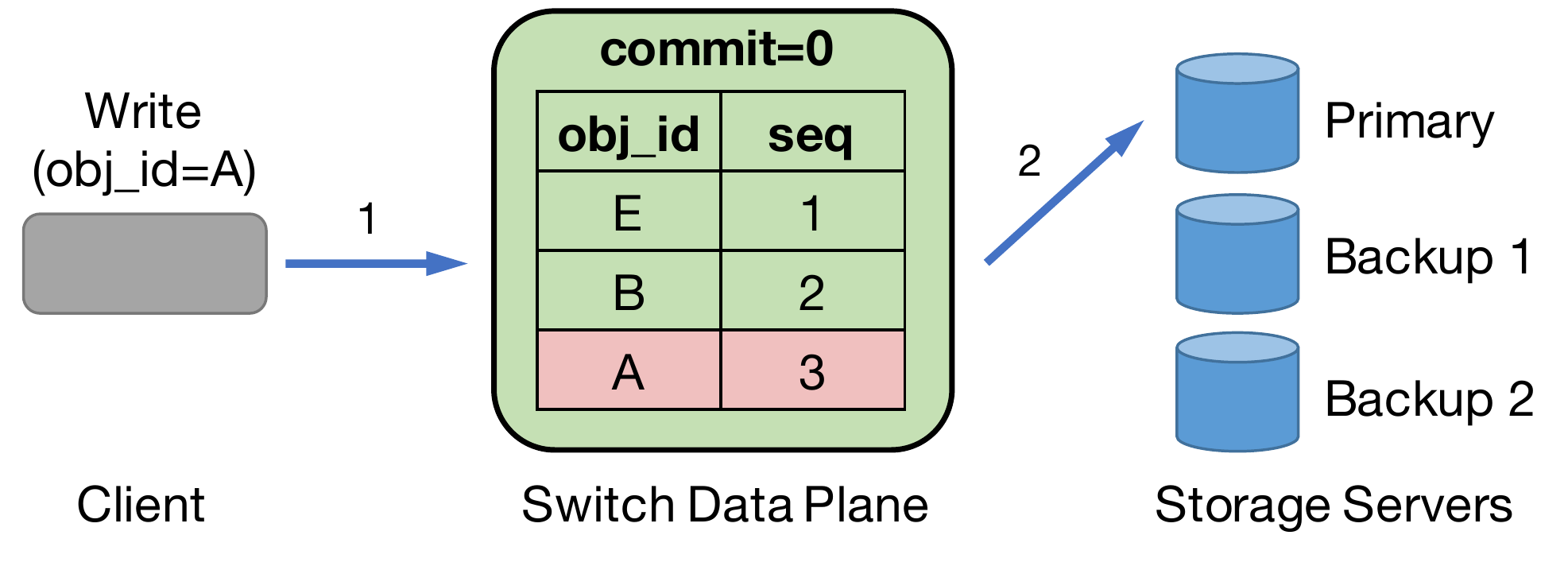}}
    \subfigure[Write completion is piggybacked in write reply.]{
        \label{fig:design_query2}
        \includegraphics[width=0.45\linewidth]{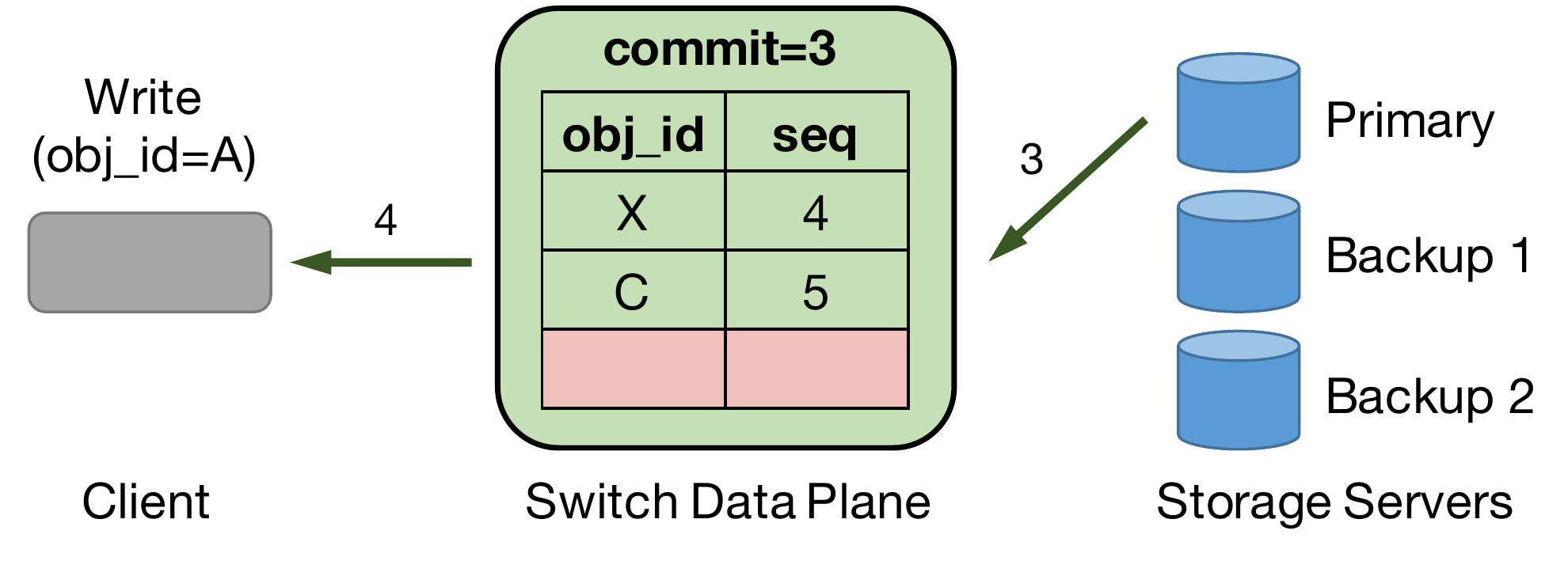}}
    \subfigure[Read and reply on object with pending writes.]{
        \label{fig:design_query3}
        \includegraphics[width=0.45\linewidth]{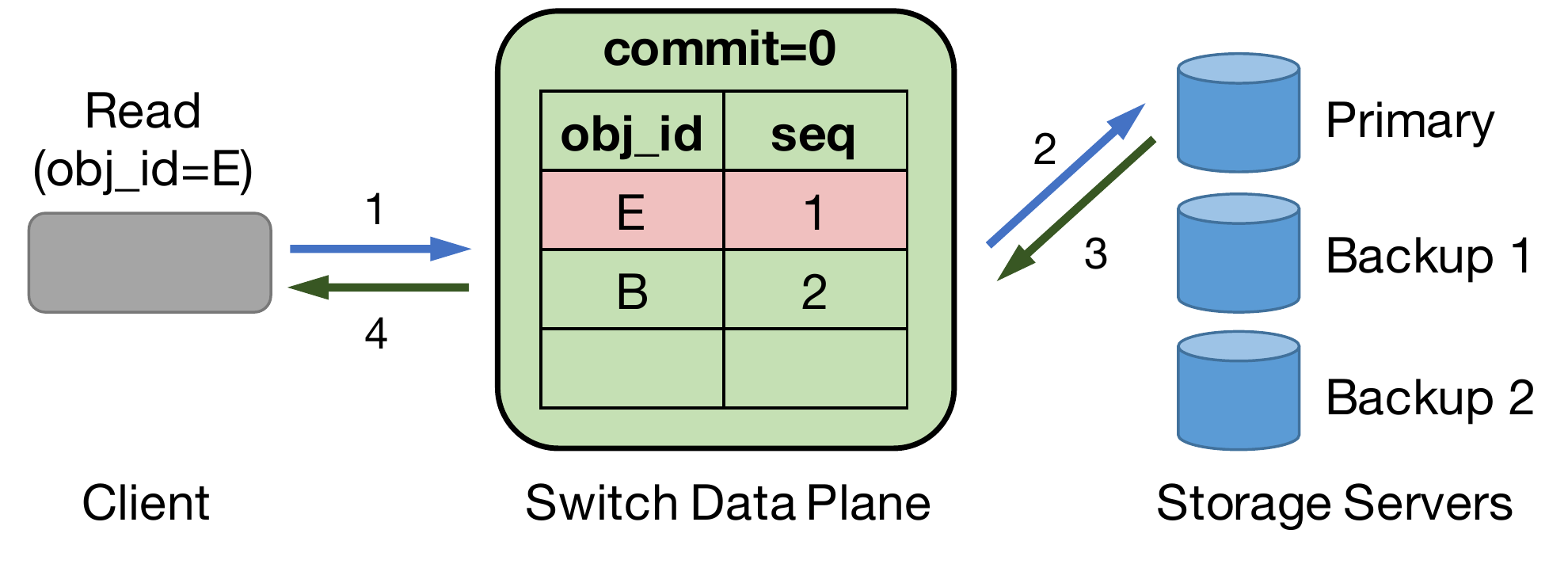}}
    \subfigure[Read and reply on object without pending writes.]{
        \label{fig:design_query4}
        \includegraphics[width=0.45\linewidth]{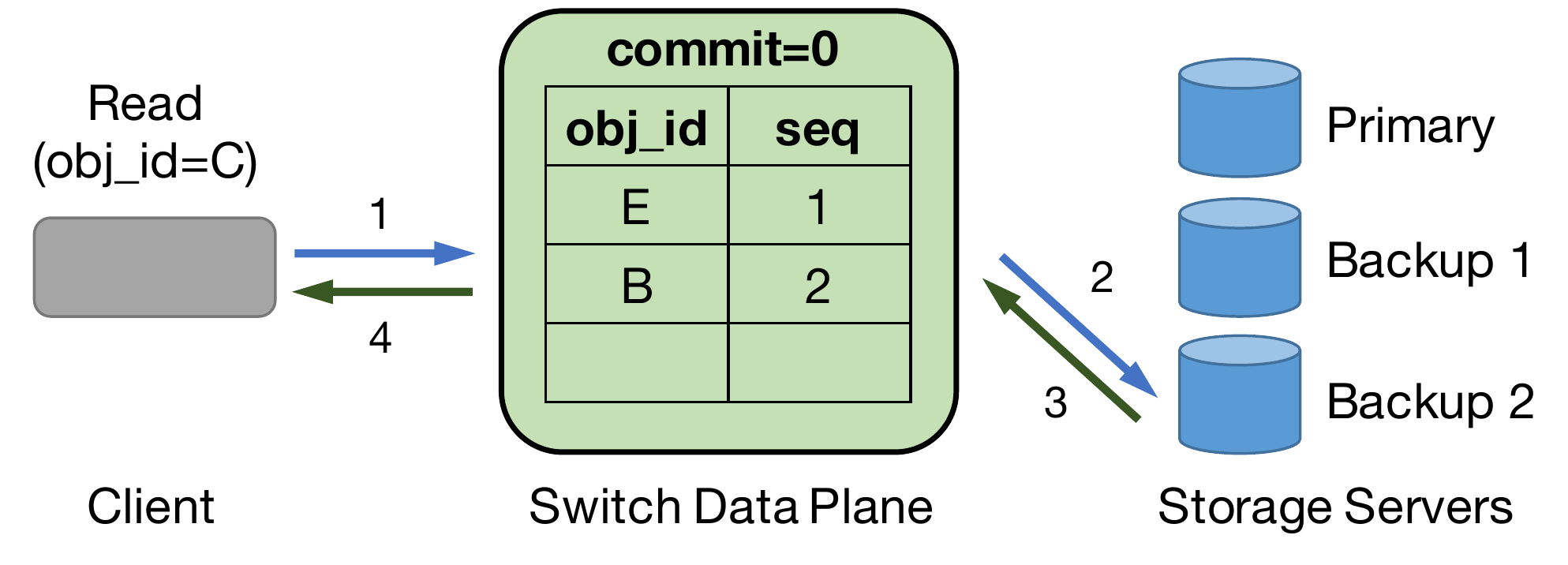}}
    \vspace{-0.1in}
    \caption{Handling different types of storage requests.}
    \vspace{-0.1in}
    \label{fig:design_query}
\end{figure*}

\subsection{Basic Request Processing}
\label{sec:design:request}

The \sysname in-switch request scheduler processes three types of operations:
\msg{read}, \msg{write}, and \msg{write-completion}. For each replicated system,
the switch is initialized with the addresses of the replicas and tracks the
three pieces of state described above: the dirty set, the monotonically-increasing
sequence number, and the sequence number of the latest committed
write. The handling of a single request is outlined in pseudo code
in Algorithm~\ref{alg:switch}.

\para{Writes.} All writes are assigned a sequence number
by \sysname. The objects being written are added to the dirty set in the switch, and
associated with the sequence number assigned to the write (lines 1--4).

\para{Write completions.} Write completions are special messages sent by the
replication protocol once a write is fully committed. If a write is the last
outstanding write to the object, the object is removed from the dirty set in the
switch. The last-committed sequence number maintained by the switch is then
updated (lines 5--8).

\para{Reads.} Reads are routed by the switch, either
through the normal replication protocol or to a randomly selected replica, based
on whether the object being read is contended or not. The switch
checks whether the ID of the object being read is in the dirty set. If
so, the switch sends the request unmodified, causing it to be processed
according to the normal replication protocol; otherwise, the read
is sent to a random replica for better performance (lines
9--12). The request is also stamped with the last-committed sequence number on the
switch for linearizability, as will be discussed in \S\ref{sec:design:reordering}.

\para{Example.} Figure~\ref{fig:design_query} shows an example workflow.
Figure~\ref{fig:design_query1} and \ref{fig:design_query2} show a
write operation. The switch adds \texttt{obj\_ID=A} to
the dirty set when it observes the write. It removes the object ID
upon the write completion, which can be piggybacked in the write reply, and updates
the last-committed sequence number. What is in the dirty set determines
how reads are handled. In Figure~\ref{fig:design_query3}, the
read is for object \texttt{E}, which has pending writes, so
the request is sent to the primary for guaranteeing consistency. On the other
hand, in Figure~\ref{fig:design_query4}, object \texttt{C} is not in
the dirty set, so it is sent to the second backup for better
performance.

\subsection{Handling Network Asynchrony}
\label{sec:design:reordering}

In an ideal network, where messages are processed in order, only using
the dirty set would be sufficient to guarantee consistency.
In a real, asynchronous network, just
because a read to an object was uncontended when the request passed
through the switch does not mean it will still be so when the request
arrives at a replica: the message might have been delayed so long that
a new write to the same object has been partially processed. \sysname
avoids this using the sequence number and last-committed point.

\para{Write order requirement.} The key invariant of the dirty set
requires that an object not be removed until \emph{all} concurrent
writes to that object have completed. Since the \sysname switch only
keeps track of the largest sequence number for each object, \sysname
requires that the replication protocol processes writes only in sequence
number order. This is straightforward to implement in a replication
protocol, e.g., via tracking the last received sequence
number and discarding any out-of-order writes.

\para{Dropped messages.} If a \msg{write-completion} or forwarded
\msg{write} message is dropped, an object may remain in the dirty set
indefinitely. While in principle harmless---it is always safe to
consider an uncontended object dirty---it may cause a performance
degradation. However, because writes are processed in order, any stray
entries in the dirty set can be removed as soon as a
\msg{write-completion} message with a higher sequence number arrives.
These stray objects can be removed by the switch as it processes reads
(i.e., by removing the object ID if its sequence number in the dirty
set is less than or equal to the last committed sequence number). This
removal can also be done periodically.

\begin{figure*}[t]
\centering
    \includegraphics[width=\textwidth]{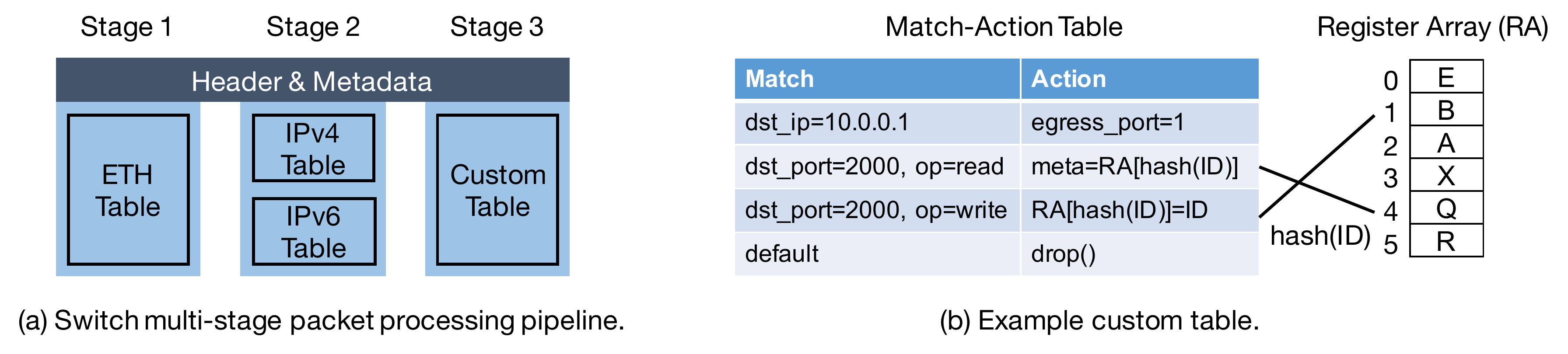}
\vspace{-0.3in}
\caption{Switch data plane structure.}
\vspace{-0.1in}
\label{fig:design_switch}
\end{figure*}

\para{Last-committed point for linearizability.} \sysname aims to \emph{maintain
linearizability}, even when the network can arbitrarily delay or
reorder packets. The switch uses its dirty set to ensure that a
single-replica read does not contend with ongoing writes \emph{at the
time it is processed by the switch}. This is not sufficient
to entirely eliminate inconsistent reads. However,
the last-committed sequence number stamped into the read will provide enough information
for the recipient to compute whether or not processing the read
locally is safe. In the primary-backup, a write after a read on the same object
would have a \emph{higher} sequence number than the last-committed point carried in the read.
As such, a backup can detect the conflict even if the write
happens to be executed at the backup before the read arrives, and then send the read
to the primary for linearizability. Detailed discussion on adapting
protocols is presented in \S\ref{sec:protocols}.

\subsection{Failure Handling}
\label{sec:design:failure}


\sysname would be of limited utility to replication protocols if the
switch were a single point of failure. However, because the switch
only keeps \emph{soft state} (i.e., the dirty set, the sequence
number and the last-committed point), it can be rebooted or replaced.
The failure of a switch will
result in temporary performance degradation.
The \sysname failure handling protocol restores the ability for the new switch
to send writes and reads through the normal case first, and then restores the
single-replica read capability, limiting the downtime of the system to a
minimum. As such, the switch is not a single point of failure, and can be safely
replaced without sacrificing consistency.

\para{Handling switch failures.} When the switch fails, the operator
either reboots it or replaces it with a backup switch. While the
switch only maintains soft state, care must be taken to preserve
consistency during the transition. As in prior
systems~\cite{nopaxos,eris}, the sequence numbers in \sysname are
augmented with the switch's unique ID and ordered lexicographically
considering the switch's ID first. This ensures that no two writes
have the same sequence number. Next, before a newly initialized switch
can process writes, \sysname must guarantee that single-replica reads
issues by the previous switch will not be processed. Otherwise, in
read-behind protocols, the previous switch and a lagging replica could
bilaterally process a read without seeing the results of the latest
writes, resulting in read-behind anomalies. To prevent these
anomalies, the replication protocol periodically agrees to allow
single-replica reads from the current switch for a time period. Before
allowing the new switch to send writes, the replication must agree to
refuse single-replica reads with smaller switch IDs, and either the
previous switch's time should expire or all replicas should agree to
cut it short. This technique is similar in spirit to the leases used
as an optimization to allow leader-only reads in many protocols.
Finally, once the switch receives its first \msg{write-completion}
with the new switch ID, both its last-committed point and dirty set
will be up to date, and it can safely send single-replica reads.

\para{Handling server failures.} The storage system handles a server failure
based on the replication protocol, and notifies the switch control plane at the
beginning and end of the process. The switch first removes the failed replica
from the replica addresses in the data plane, so that following requests would
not be scheduled to it. After the failed replica is recovered or a replacement
server is added, the switch adds the corresponding address to the replica
addresses, enabling requests to be scheduled to the server.


\renewcommand{\secname}{Data Plane Design and Implementation}
\section{\secname}
\label{sec:switch}

Can the \sysname approach be supported by a real
switch? We answer this in the affirmative by showing how to implement it
for a programmable switch~\cite{rmt, p4-ccr} (e.g., Barefoot's
Tofino~\cite{tofino}), and evaluate its resource usage.

\subsection{Data Plane Design}
\label{sec:switch:design}

The in-network conflict detection module is implemented in the data plane of a
modern programmable switch. The sequence number and last-committed point can be stored
in two registers, and the dirty set can be stored in a hash table implemented
with register arrays. While previous work has shown how to use the register
arrays to store key-value data in a switch~\cite{netcache, netchain}, our work
has two major differences: $(i)$ the hash table only needs to store object IDs,
instead of both IDs and values; $(ii)$ the hash table needs to support
insertion, search and deletion operations at line rate, instead of only search.
We provide some background on programmable switches, and then describe
the hash table design.

\begin{figure*}[t]
\centering
    \includegraphics[width=\textwidth]{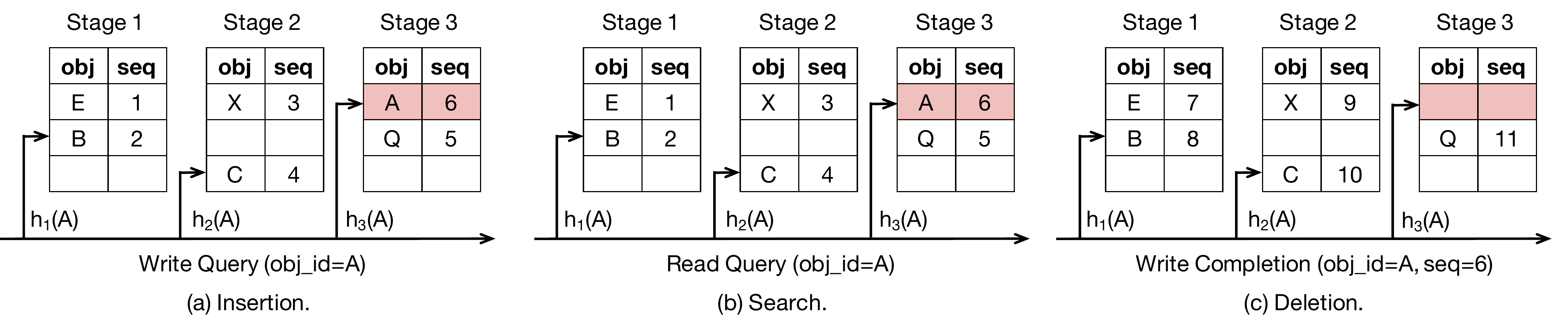}
\vspace{-0.25in}
\caption{Multi-stage hash table design that supports insertion, search and deletion
in the switch data plane.}
\vspace{-0.1in}
\label{fig:design_hash_detail}
\end{figure*}

\para{Switch data plane structure.} Figure~\ref{fig:design_switch}
illustrates the basic data plane structure of a modern programmable switching ASIC.
The packet processing pipeline contains multiple stages, as shown in
Figure~\ref{fig:design_switch}(a). Packets are processed by the stages one after
one. Match-action tables are the basic element used to process packets.
If two tables have no dependencies, they can be placed in the same
stage, e.g., IPv4 and IPv6 tables in
Figure~\ref{fig:design_switch}(a).

A match-action table contains a list of rules that specifies how packets are
processed, as shown in Figure~\ref{fig:design_switch}(b). A match
in a rule specifies a header pattern, and the action specifies how the matched packets should
be processed. For example, the first rule in Figure~\ref{fig:design_switch}(b)
forwards packets to egress port 1 for packets with destination IP 10.0.0.1. Each
stage also contains register arrays that can be accessed at line rate.
Programmable switches allow developers to define custom packet formats and
match-action tables to realize their own protocols. The example in
Figure~\ref{fig:design_switch}(b) assumes two custom fields in the packet
header, which are \texttt{op} for operation and \texttt{ID} for object ID. The
second and third rules perform read and write on the register array based on
\texttt{op} type, and the index of the register array is computed by the hash of
\texttt{ID}.

Developers use a domain-specific language such as P4~\cite{p4-ccr} to write a
program for a custom data plane, and then use a complier to
compile the program to a binary that can be loaded to the switch. Each stage has
resource constraints on the size of match-action tables (depending on the
complexity of matches and actions) and register arrays (depending on the length
and width).

\para{Multi-stage hash table with register arrays.} The switch data
plane provides basic constructs for the conflict
detection module. A register array can be naturally used to store the object
IDs. We can use the hash of an object ID
as the index of the register array, and store the object ID as the value in the
register slot. One challenge is
to handle hash collisions, as the switch can only perform a limited, fixed number
of operations per stage. Collision resolution for hash tables is a well-studied
problem, and the multi-stage structure of the switch data plane makes it natural to
implement open addressing techniques to handle collisions. Specifically, we allocate
a register array in each stage and use different hash functions for different stages.
In this way, if several objects collide in one stage, they are less likely to collide
in another stage.
Figure~\ref{fig:design_hash_detail} shows the design.
\begin{itemize}[leftmargin=*]
    \item \textbf{Insertion.} For a write, the object ID is inserted to
    the first stage with an empty slot for the object
    (Figure~\ref{fig:design_hash_detail}(a)). The write is dropped if no slot
    is available.

    \item \textbf{Search.} For a read, the switch iterates over all stages
    to see if any slot contains the same object ID
    (Figure~\ref{fig:design_hash_detail}(b)).

    \item \textbf{Deletion.} For a write completion, the switch iterates over all
    stages and removes the object ID
    (Figure~\ref{fig:design_hash_detail}(c)).
\end{itemize}

\para{Variable-length object IDs.} Many systems use
\emph{variable-length} IDs,
but due to switch limitations, \sysname must use \emph{fixed-length}
object IDs for conflict detection. However, variable-length IDs can be
accommodated by having the clients store fixed-length hashes of the
original ID in the \sysname packet header; the original ID is sent in
the packet payload. \sysname then uses the fixed-length hashes for
conflict detection. Hash collisions may degrade performance but
cannot introduce consistency issues; they can only cause \sysname to
believe a key is contended, not vice versa.

\subsection{Resource Usage}
\label{sec:switch:resource}

Switch on-chip memory is a limited resource. Will there be enough memory to
store the entire dirty set of pending writes? Our key insight is that
since the switch only performs conflict detection, it does not need to
store actual data, but only the object IDs. This is in
contrast to previous designs like NetCache~\cite{netcache} and NetChain~\cite{netchain}
that use switch memory for object storage directly. Moreover, while the storage
system can store a massive number of objects, the number of writes at any given time is
small, implying that the dirty set is far smaller than the storage size.

Suppose we use $n$ stages and each stage has a register array with $m$ slots.
Let the hash table utilization be $u$ to account for hash collisions. The switch
is able to support up to $unm$ writes at a given time.  Suppose the duration of
each write is $t$, and the write ratio is $w$. Then the switch is able to
support $unm/t$ writes per second---or a total throughput of $unm/(wt)$---before
exhausting memory. As a concrete example, let $n$=3, $m$=64000, $u$=50\%, $t$=1
ms and $w$=5\%.
The switch can support a write throughput of 96 million requests per
second (MRPS), and a total throughput of 1.92 billion requests per
second (BRPS). Let both the object ID and sequence number be 32 bits.
It only consumes 1.5MB memory. Given that a commodity switch has
10--20 stages and a few tens of MB memory~\cite{flexswitch, netcache, netchain},
this example only conservatively uses a small fraction of switch memory.

\subsection{Deployment Issues}
\label{sec:switch:issue}

We imagine two possible deployment scenarios for \sysname. First, it
can be easily integrated with clustered storage systems, such as
on-premise storage clusters for enterprises and specialized storage
clusters in the cloud. As shown in
Figure~\ref{fig:overview_architecture}, all servers are
deployed in the same rack, allowing the ToR switch to be the central
location that sees all the storage traffic. We only need to add
\sysname's functionality to the ToR switch.

For cloud-scale storage, replicas may be distributed among
many different racks for fault tolerance. Placing the \sysname
scheduler on a ToR switch, which only sees storage traffic to its own
rack, does not suffice. Instead, we leverage a network serialization
approach~\cite{nopaxos,specpaxos}, where all traffic destined for a
replica group is redirected through a designated switch. Prior work
has shown that, with careful selection of the switch (e.g., a spine
switch in a two-layer leaf-spine network), this need not increase
latency~\cite{nopaxos}. Nor does it impose a throughput penalty:
different replica groups can use different switches as their request
scheduler, and the capacity of a switch far exceeds that of a single
replica group.


\renewcommand{\secname}{Adapting Replication Protocols}
\section{\secname}
\label{sec:protocols}

Safely using a replication protocol with \sysname imposes three
responsibilities on the protocol. It must:

\begin{enumerate}
\item process writes only in sequence number order;
\item allow single-replica reads only from one active switch at a time; and
\item ensure that single-replica reads for uncontended objects still return
  linearizable results.
\end{enumerate}

Responsibility (1) can be handled trivially by dropping messages that
arrive out of order, and responsibility (2) can be implemented in the
same manner as leader leases in traditional replication protocols. We
therefore focus on responsibility (3) here. How this is handled is
different for the two categories of read-ahead and read-behind protocols.

To demonstrate the generality of our approach, we apply \sysname to
representative protocols from both classes: primary-backup protocols
(including chain replication), as well as leader-based quorum
protocols (Viewstamped Replication/Multi-Paxos) and recent
single-phase consensus protocols (NOPaxos). For each, we explain the
necessary protocol modifications and give a brief argument for
correctness. A full proof of correctness is in
Appendix~\ref{app:proof}, and a model-checked TLA+ specification of
\sysname is in Appendix~\ref{app:tla}.

\subsection{Requirements for Linearizability}

Let us first specify the requirements that must be satisfied for a
\sysname{}-adapted protocol to be correct. We only consider systems where the
underlying replication protocol is linearizable. All write operations are
processed by the replication protocol based on the sequence number order.
We need only, then, consider the read operations.
The following two properties are sufficient for linearizability.

\begin{itemize}[leftmargin=*]
\item \textbf{P1. Visibility.} A read operation sees the effects of all
  write operations that finished before it started.
\item \textbf{P2. Integrity.} A read operation will not see the
  effects of any write operation that had not committed at the time
  the read finished.
\end{itemize}

In the context of \sysname, read operations follow the normal-case
replication protocol if they refer to an object in the dirty set, and
hence we need only consider the fast-path read operations executed at a
single replica. For these, P1 can equivalently be stated as follows.
\begin{itemize}[leftmargin=*]
\item \textbf{P1. Visibility}.
  The replication protocol must only send a completion notification
  for a write to the scheduler if any subsequent single-replica read
  sent to any replica will reflect the effect of the write operation.
\end{itemize}

\noindent

\subsection{Read-Ahead Protocols}

Both primary-backup and chain replication are read-ahead protocols
that cannot have read-behind anomalies, because they only reply to
the client once an operation has been executed on all replicas. As a
result, they inherently satisfy P1. We adapt them to send a
\msg{write-completion} notification to the switch at the same time as
responding to the client.

However, read-ahead anomalies \emph{are} possible: reads naively executed at
a single replica can reflect uncommitted results. We use the
last-committed sequence number provided by the \sysname switch to
prevent this. When a replica receives a fast-path read for object $o$,
it checks that the last-committed sequence number attached to the
request is at least as large as the latest write applied to $o$. If
not, it forwards the request to the primary or tail, to be executed
using the normal protocol. Otherwise, this implies that all writes to
$o$ processed by the replica were committed at the time the read was
handled by the switch, satisfying P2.

\subsection{Read-Behind Protocols}

We have applied \sysname to two quorum protocols: Viewstamped
Replication~\cite{viewstamped,liskov12:_views_replic_revis}, a
leader-based consensus protocol equivalent to
Multi-Paxos~\cite{paxos01} or Raft~\cite{raft}, and NOPaxos~\cite{nopaxos}, a
network-aware, single-phase consensus protocol. Both are read-behind
protocols. Because replicas in these protocols only execute operations
once they have been committed, P2 is trivially satisfied.

Furthermore, because the last committed point in the \sysname switch
is greater than or equal to the sequence numbers of all writes removed
from its dirty set, replicas can ensure visibility (P1) by rejecting
(and sending to the leader for processing through the normal protocol)
all fast-path reads whose last committed points are larger than that of
the last locally \emph{committed and executed} write.

In read-behind protocols, \msg{write-completion}s can be sent along
with the response to the client. However, in order to reduce the
number of rejected fast-path reads, we delay \msg{write-completion}s
until the write has likely been executed on all replicas.


\para{Viewstamped replication.}
For Viewstamped Replication,
we add an additional phase to operation processing that ensures a
quorum of replicas have \emph{committed and executed} the operation.
Concurrently with responding to the client, the VR leader sends a
\msg{commit} message to the other replicas. Our additional phase calls
for the replicas to respond with a \msg{commit-ack} message.\footnote{%
  These messages can be piggybacked on the next \msg{prepare} and
  \msg{prepare-ok} messages, eliminating overhead.
} Only once the leader receives a quorum of \msg{commit-ack} messages
for an operation with sequence number $n$ does it send a
$\langle\msg{write-completion}, \textit{object\_id}, n\rangle$
notification.


\para{NOPaxos.}
NOPaxos~\cite{nopaxos} uses an in-network sequencer to enable a
single-round, coordination-free consensus protocol. It is a natural
fit for \sysname, as both the sequencer and \sysname's request
scheduler can be deployed in the same switch. Although NOPaxos
replicas do not coordinate while handling operations, they already run
a periodic synchronization protocol to ensure that all replicas have
executed a common, consistent prefix of the log~\cite{nopaxos-tr} that
serves the same purpose as the additional phase in VR. The only
\sysname modification needed is for the leader, upon completion of a
synchronization, to send
$\langle\msg{write-completion}, \textit{object\_id}, commit\rangle$
messages for all affected objects.

\renewcommand{\secname}{Implementation}
\section{\secname}
\label{sec:implementation}

We have implemented a \sysname prototype and integrated it with Redis~\cite{redis}.
The switch data plane is implemented in P4~\cite{p4-ccr} and is compiled to
Barefoot Tofino ASIC~\cite{tofino} with Barefoot Capilano software suite~\cite{Capilano}. We
use 32-bit object IDs, and use 3 stages for the hash table.
Each stage provides 64K slots to store the object IDs, resulting
in a total of 192K slots for the hash table.

The shim layer in the
storage servers is implemented in C++. It communicates with clients using
\sysname packets, and uses hiredis~\cite{hiredis}, which is the official C
library of Redis~\cite{redis}, to read from and write to Redis. In additional to
translate between \sysname packets and Redis operations, the shim layers in the
servers also communicate with each other to implement replication protocols. We have
integrated \sysname with multiple representative replication protocols
(\S\ref{sec:evaluation:generality}). We
use the pipeline feature of Redis to batch requests to Redis.
Because Redis is single-threaded, we run eight Redis
processes on each server to maximize per-server throughput.
Our prototype
is able to achieve about 0.92 MQPS for reads and 0.8 MQPS for writes on a single
Redis server.

The client library is implemented in C. It generates
mixed read and write requests to the storage system, and
measures the system throughput and latency.

\renewcommand{\secname}{Evaluation}
\section{\secname}
\label{sec:evaluation}

\begin{figure}[t]
    \centering
    \subfigure[Read-only workload.]{
        \label{fig:eval_read_latency}
        \includegraphics[width=0.48\linewidth]{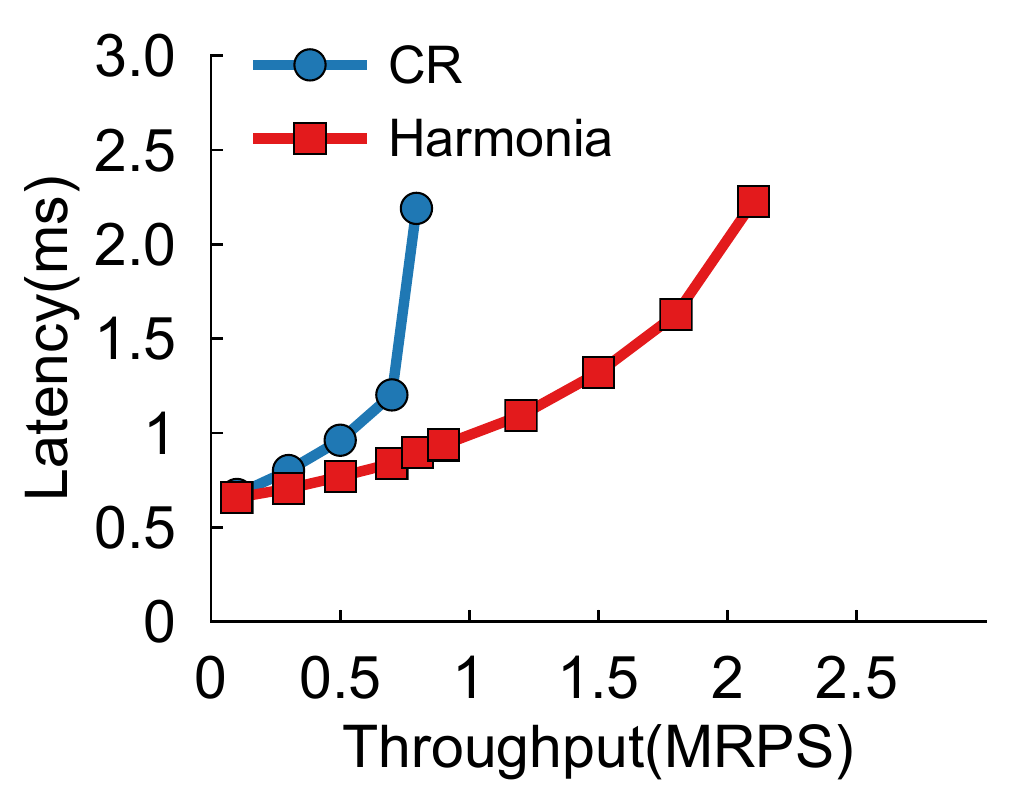}}
    \subfigure[Write-only workload.]{
        \label{fig:eval_write_latency}
        \includegraphics[width=0.48\linewidth]{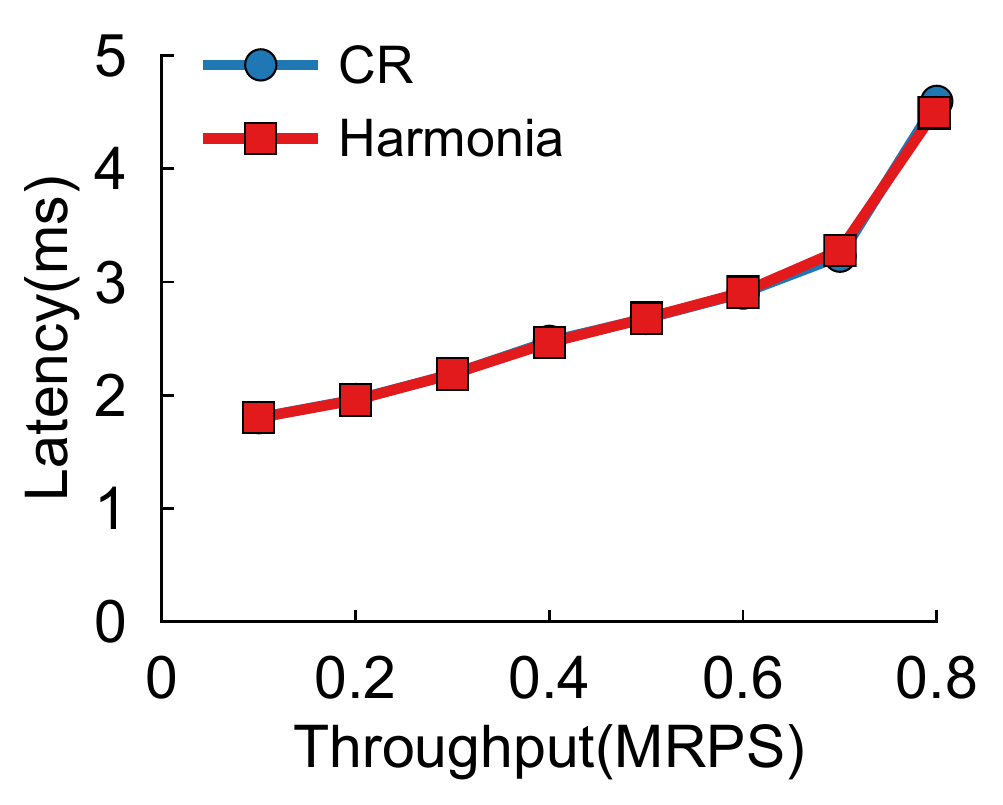}}
    \vspace{-0.1in}
    \caption{Throughput vs. latency for reads and writes.}
    \label{fig:eval_throughput_latency}
\end{figure}

We provide experimental results to demonstrate that \sysname provides
significant throughput and latency improvements (\S\ref{sec:evaluation:performance}),
scales out with the number of replicas (\S\ref{sec:evaluation:scalability}),
is resource efficient (\S\ref{sec:evaluation:memory}),
is general to many replication protocols (\S\ref{sec:evaluation:generality}),
and handles failures gracefully (\S\ref{sec:evaluation:failure}).

\subsection{Methodology}
\label{sec:evaluation:methodology}

\para{Testbed.} Our experiments are conducted on a testbed consisting of twelve
server machines connected by a 6.5 Tbps Barefoot Tofino switch. Each server is
equipped with an 8-core CPU (Intel Xeon E5-2620 @ 2.1GHz), 64 GB total memory,
and one 40G NIC (Intel XL710). The server OS is Ubuntu 16.04.3 LTS. Ten
storage servers run Redis v4.0.6~\cite{redis} as the storage backend; two
generate client load using a DPDK-based workload generator.

By default, we use three replicas and a uniform workload on one million objects
with 5\% write ratio. The 5\% write ratio is similar to that in many
real-world storage systems~\cite{memcache-nsdi13, workload-fb-sigmetrics12}, previous
studies~\cite{incbricks-asplos17}, and standard benchmarks like
YCSB~\cite{ycsb}. We vary the parameters in the experiments to evaluate their
impacts.

\para{Comparison.} Redis is a widely-used open-source in-memory storage system.
However, Redis does not provide native support for replication, only a
cluster mode with weak consistency. We use a shim layer to implement several representative replication
protocols, including primary-backup (PB)~\cite{primary-backup},
chain replication (CR)~\cite{chain-replication},
CRAQ~\cite{craq} (a version of chain replication that makes reads more
scalable at the cost of more expensive writes), Viewstamped
Replication (VR)~\cite{viewstamped} and NOPaxos~\cite{nopaxos}. As described
in \S\ref{sec:implementation}, we run eight Redis
processes on each server to maximize per-server throughput. The shim layer
batches requests to Redis; the baseline (unreplicated) performance for one
server is 0.92 MQPS for reads and 0.8 MQPS for writes.

We compare system performance with and without \sysname for each
protocol. Due to space constraints, we show the results of CR, which is a high-throughput variant of PB, in most
figures; \S\ref{sec:evaluation:generality} compares performance across
all protocols, demonstrating the generality.

\begin{figure}[t]
    \centering
    \subfigure[Read vs. write throughput.]{
        \label{fig:eval_read_write}
        \includegraphics[width=0.48\linewidth]{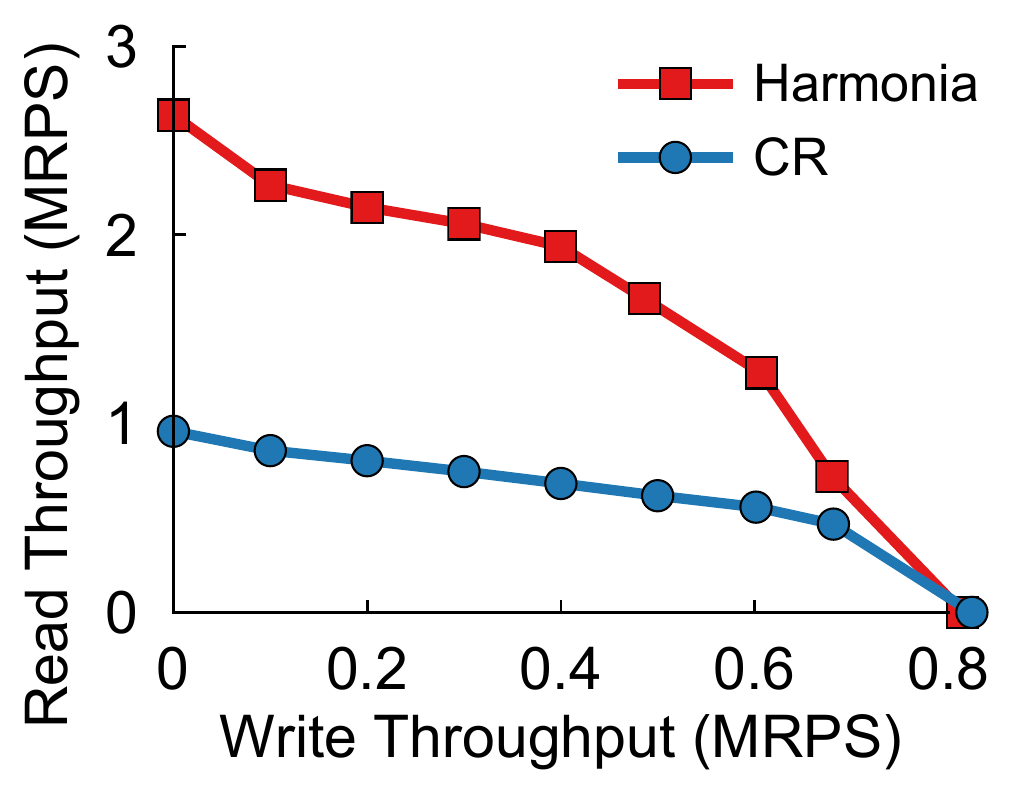}}
    \subfigure[Throughput vs. write ratio.]{
        \label{fig:eval_throughput_ratio}
        \includegraphics[width=0.48\linewidth]{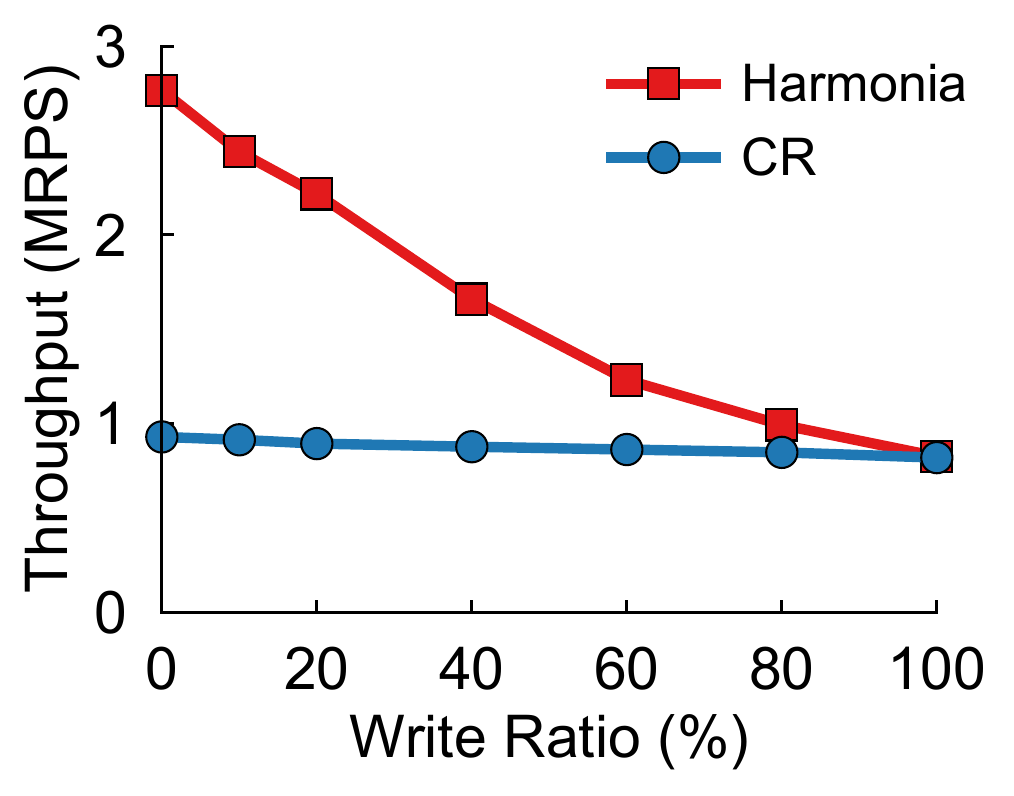}}
    \vspace{-0.1in}
    \caption{Throughput for mixed read-write workloads. (a) Read throughput as the write
    rate increases. (b) Total throughput under different write ratios.}
    \label{fig:eval_mixed}
\end{figure}

\begin{figure*}[t]
    \centering
    \subfigure[Read-only workload.]{
        \label{fig:eval_scalability_read}
        \includegraphics[width=0.32\linewidth]{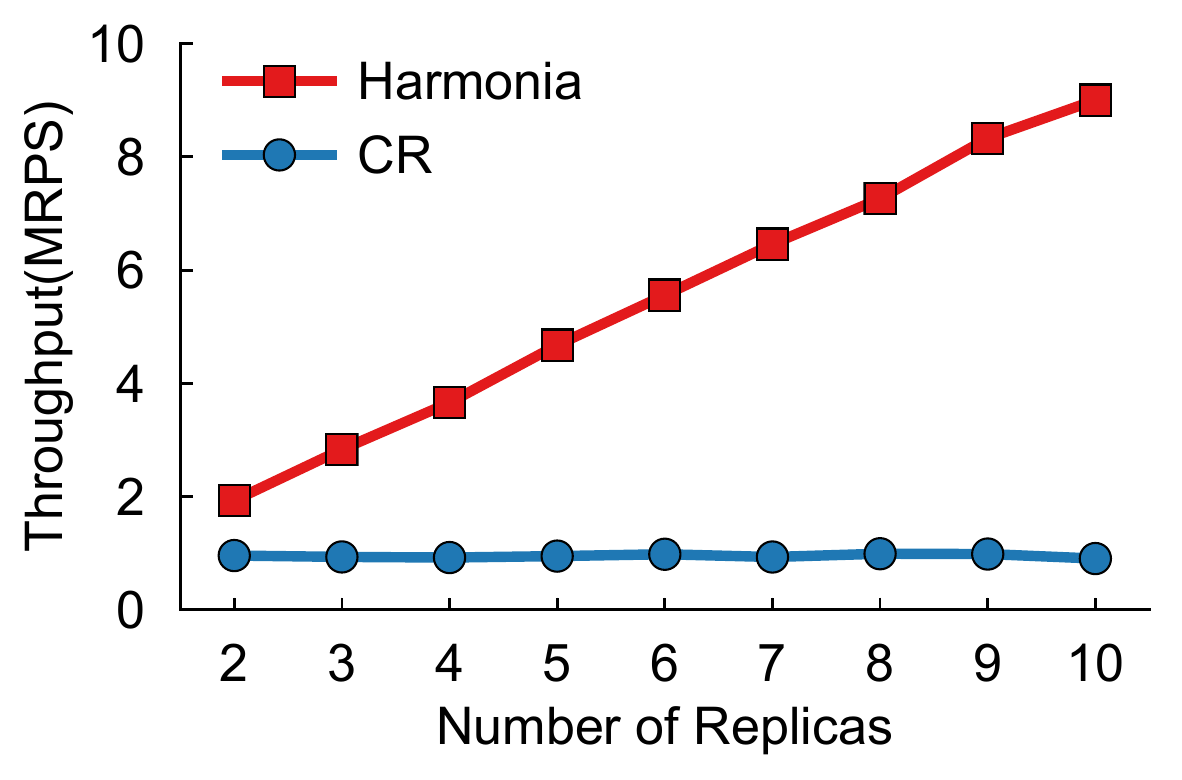}}
    \subfigure[Write-only workload.]{
        \label{fig:eval_scalability_write}
        \includegraphics[width=0.32\linewidth]{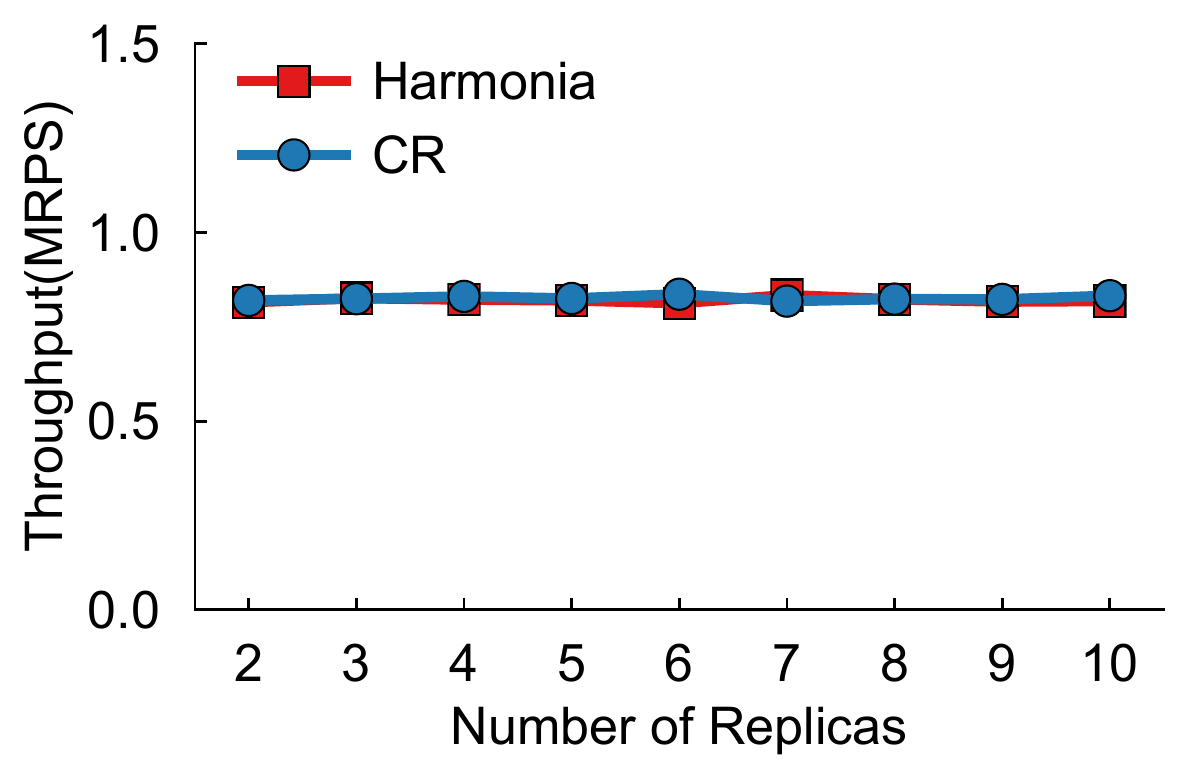}}
    \subfigure[Mixed workload with 5\% writes.]{
        \label{fig:eval_scalability_mixed}
        \includegraphics[width=0.32\linewidth]{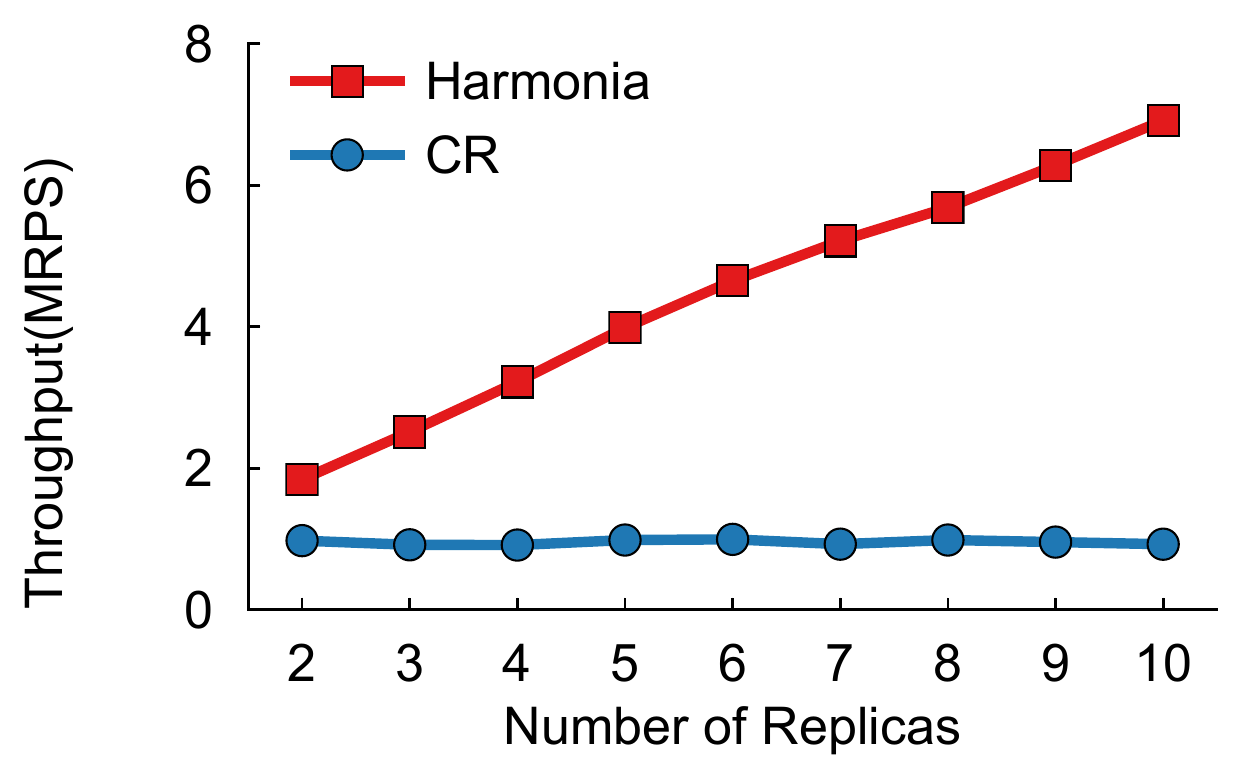}}
    \vspace{-0.05in}
    \caption{Total throughput with increasing numbers of replicas for three workloads. \sysname scales out with the number of replicas in read-only
    and read-intensive workloads.
    }
    \label{fig:eval_scalability}
\end{figure*}

\subsection{Latency vs. Throughput}
\label{sec:evaluation:performance}

We first conduct a basic throughput and latency experiment. The client generates
requests to three replicas, and measures the average latency at different
throughput levels. We consider read-only, write-only, and mixed workloads.

Figure~\ref{fig:eval_read_latency} shows the relationship between
throughput and latency under a read-only workload. Although we have
three replicas, since CR only uses the tail node to handle read
requests, the throughput is bounded by that of one server. In
comparison, since \sysname uses the switch to detect read-write
conflicts, it is able to fully utilize the capacity of all the three
replicas when there are no conflicts. The read latency is a few hundred microseconds
at low load, and increases as throughput goes up.
For write-only workloads (Figure~\ref{fig:eval_write_latency}), CR and
\sysname have identical performance, as \sysname simply passes writes
to the normal protocol.

To evaluate mixed workloads, the client fixes its rate of generating write requests,
and measures the maximum read throughput that can be handled by the replicas.
Figure~\ref{fig:eval_read_write} shows the read throughput as a function of write rate.
Since CR can only leverage the
capacity of tail node, its read throughput is no more than that of one
storage server, even when the write throughput is small. On the other hand,
\sysname can utilize almost all the three replicas to handle reads when the
write throughput is small. At low write rate, \sysname improves the throughput
by 3$\times$ over CR. At high write rate, both systems have similar throughput
as \sysname and CR process write requests in the same way.

Figure~\ref{fig:eval_throughput_ratio} evaluates the system performance for mixed
workloads from another angle. The client fixes the ratio of writes and measure
the saturated system throughput. The figure shows the total throughput as a function
of write ratio.
Similarly, the throughput of CR is bounded by the tail node,
while \sysname can leverage all replicas to process reads.
Similar to Figure~\ref{fig:eval_read_write}, when the write ratio is
high, \sysname has little benefit
as they process writes in the same way.

\subsection{Scalability}
\label{sec:evaluation:scalability}

\sysname offers near-linear read scalability for read-intensive workloads. We
demonstrate this by varying the number of replicas and measuring system
throughput in several representative cases. The scale is limited by the size of
our twelve-server testbed: we can use up to ten servers as replicas, and
two servers as clients to generate requests. Our high-performance
client implementation written in C and DPDK is able to saturate ten
replicas with two client servers.

\sysname offers dramatic improvements on read-only workloads
(Figure~\ref{fig:eval_scalability_read}). For CR, increasing the
number of replicas does not change the overall throughput, because it
only uses the tail to handle reads. In contrast, \sysname is
able to utilize the other replicas to serve reads, causing throughput
to increase linearly with the number of replicas. \sysname improves the
throughput by 10$\times$ with a replication factor of 10, limited
by the testbed size. It can scale out until the switch is
saturated. Multiple switches can be used for multiple replica groups to further scale out (\S\ref{sec:switch:issue}).

\begin{figure}[t]
\centering
    \includegraphics[width=0.75\linewidth]{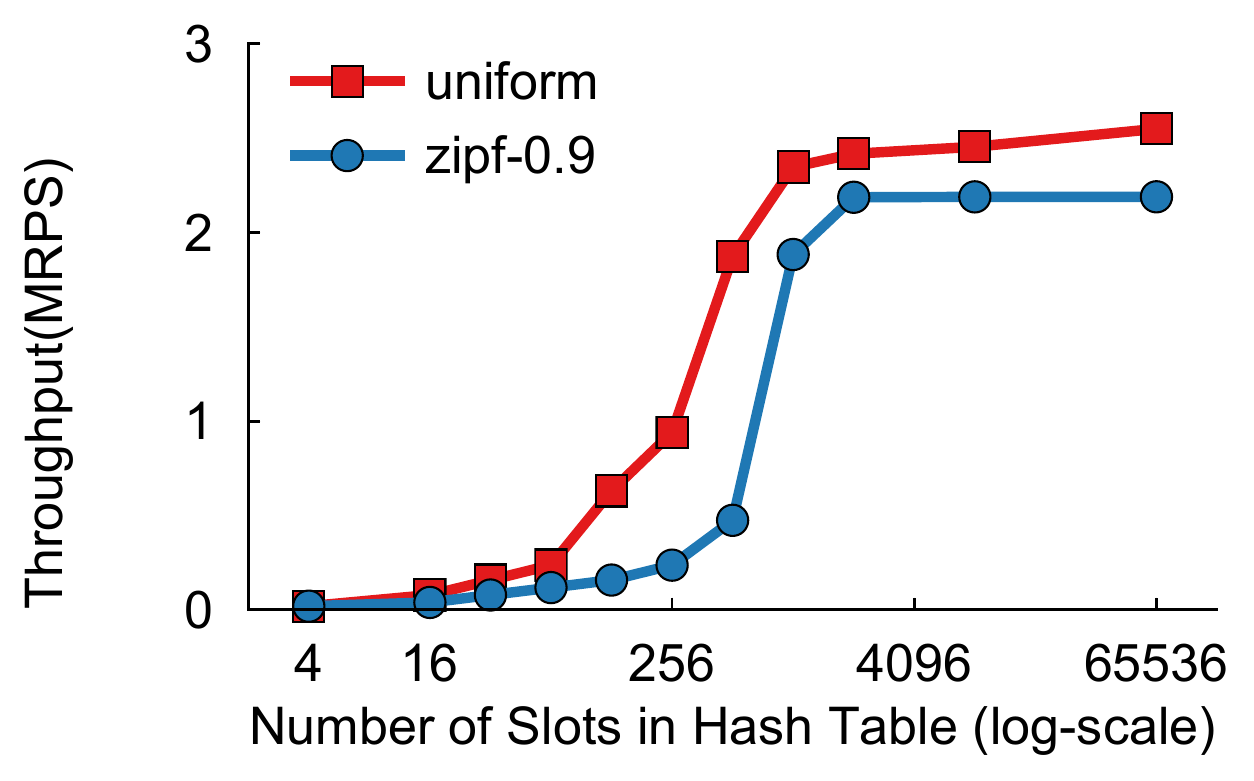}
\vspace{-0.05in}
\caption{Impact of switch memory. \sysname only consumes a small amount of
memory.}
\vspace{-0.1in}
\label{fig:eval_memory}
\end{figure}

\begin{figure*}[t]
    \centering
    \subfigure[Primary-backup protocols.]{
        \label{fig:eval_generality_pb}
        \includegraphics[width=0.45\linewidth]{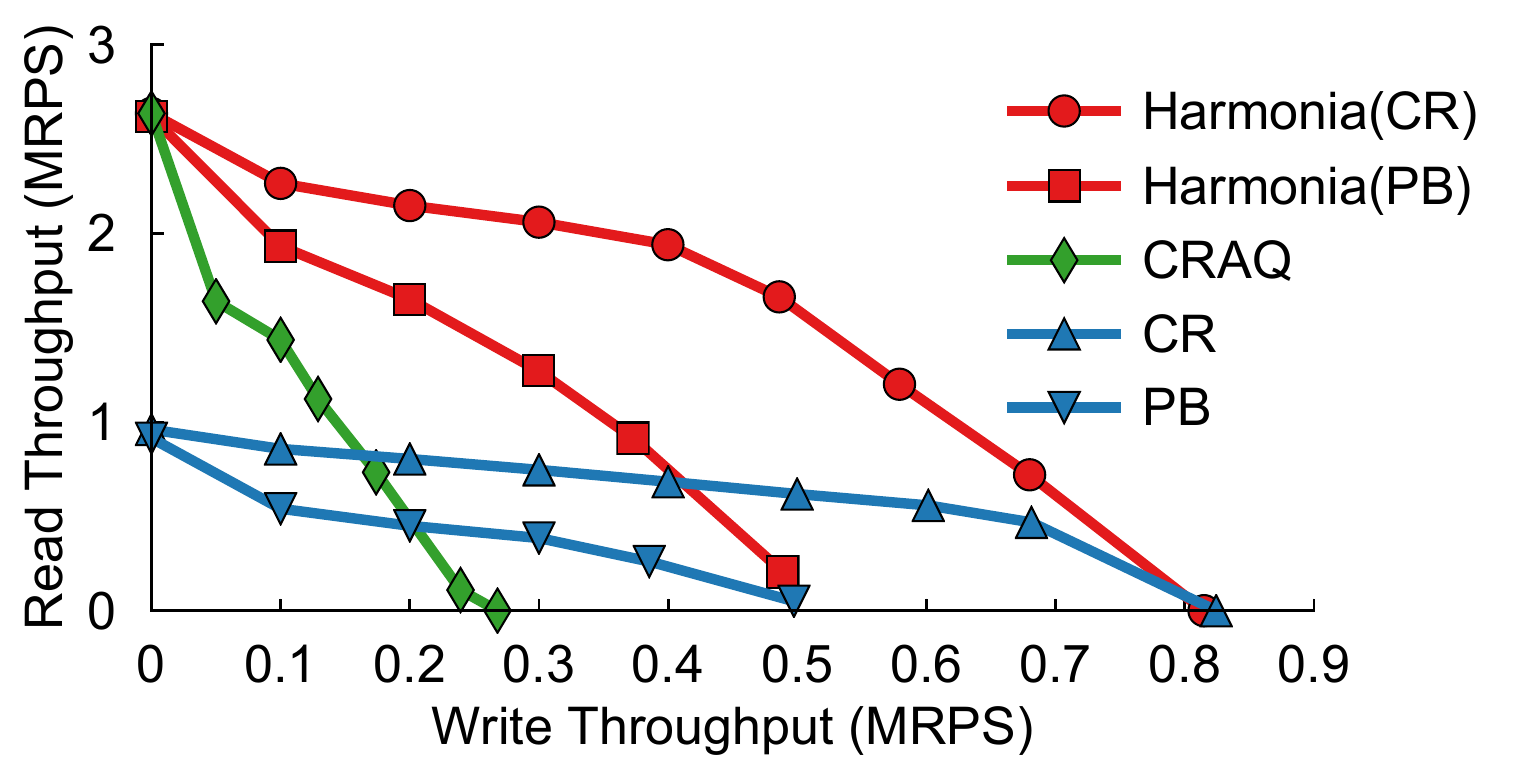}}
    \subfigure[Quorum-based protocols.]{
        \label{fig:eval_generality_quorum}
        \includegraphics[width=0.45\linewidth]{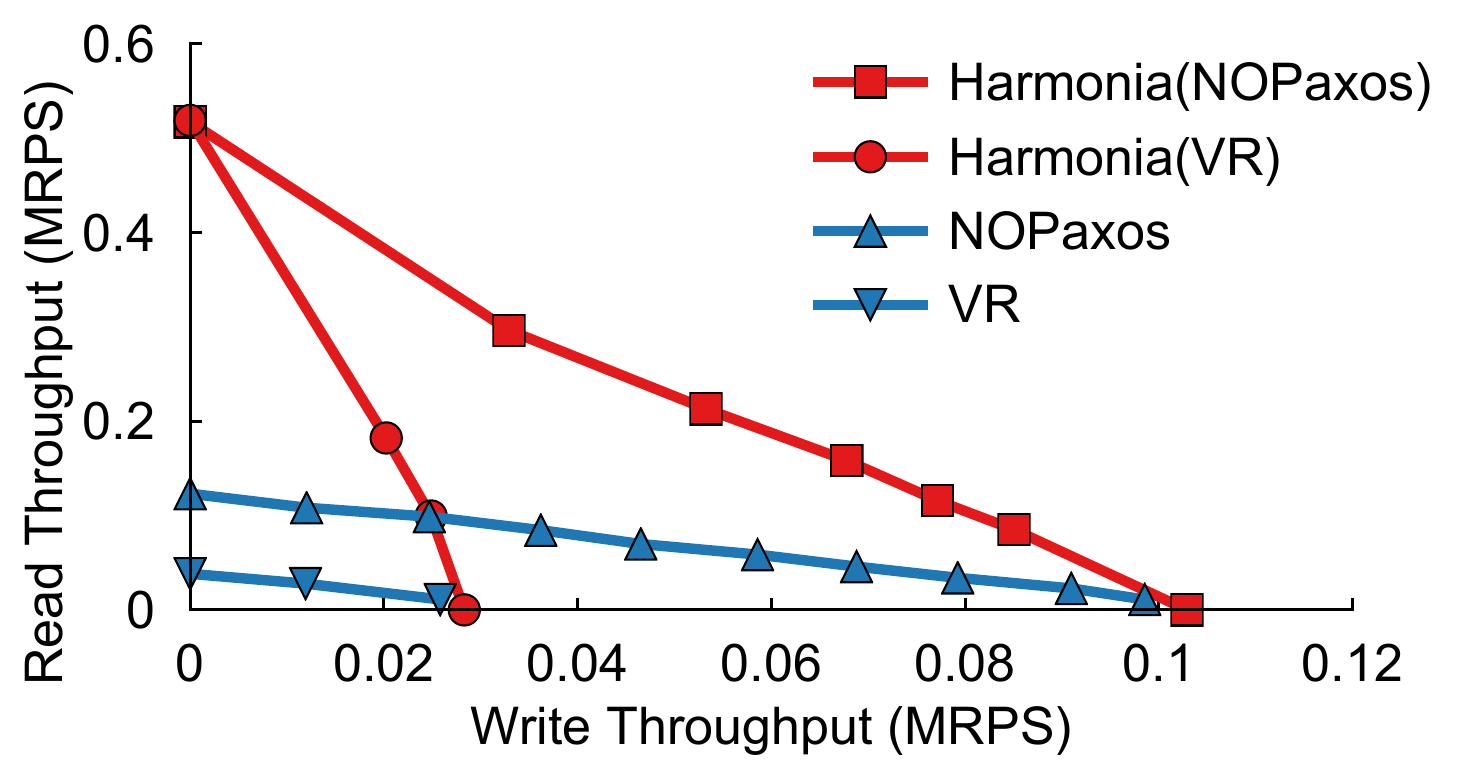}}
    \vspace{-0.1in}
    \caption{Read throughput as write rate increases, for a variety of replication
            protocols, with and without \sysname.}
    \vspace{-0.1in}
    \label{fig:eval_generality}
\end{figure*}

On write-only workloads (Figure~\ref{fig:eval_scalability_write}),
\sysname has no benefit regardless of the number of replicas used
because \sysname uses the underlying replication protocol for writes.
For CR, the throughput stays the same as more replicas are added since
CR uses a chain to propagate writes.

Figure~\ref{fig:eval_scalability_mixed} considers throughput
scalability under a mixed read-write workload with a write ratio of
5\%. Again, CR does not scale with the number of replicas. In
comparison, the throughput of \sysname increases nearly linearly with
the number of replicas. Under a read-intensive workload, \sysname can
efficiently utilize the remaining capacity on the other nodes. The
total throughput here is smaller than that for read-only requests
(Figure~\ref{fig:eval_scalability_read}), because handling writes is
more expensive than handling reads and the tail node becomes the bottleneck as the
number of replicas goes up to 8.

\subsection{Resource Usage}
\label{sec:evaluation:memory}

We now evaluate how much switch memory is needed to track the dirty
set. As we have discussed in \S\ref{sec:switch:resource}, \sysname
requires much less memory than other systems such as
NetCache~\cite{netcache} and NetChain~\cite{netchain} because \sysname
only needs to store metadata (i.e., object IDs and sequence numbers).
In this experiment, we vary the size of
\sysname switch's hash table, and measure the total throughput of
three replicas. Here, we use a write ratio of 5\% and both
uniform and skewed (zipf-0.9) request distributions across one million
keys. As shown in Figure~\ref{fig:eval_memory}, \sysname only requires
about 2000 hash table slots to track all outstanding writes
and reach maximum throughput. Before reaching the maximum, the
throughput of the uniform case increases faster than for the skewed
workload. This is because under the skewed workload, a hot object
would always occupy a slot in the hash table, making the switch drop
writes to other objects that collide on this slot, thus limiting throughput.

With 32-bit object IDs and 32-bit sequence numbers, 2000 slots only
consume 16KB memory. Given that commodity switches have tens of MB
on-chip memory~\cite{flexswitch, netcache, netchain}, the memory used by \sysname
only accounts for a tiny fraction of the total memory, e.g., only 1.6\% (0.8\%)
for 10MB (20MB) memory. This result roughly matches the back-of-envelop calculations
in \S\ref{sec:switch:resource}, with differences coming from table utilization, write duration and total throughput.
Thus, \sysname can be added to the switch and co-exist
with other modules without significant resource consumption. It also allows \sysname to scale out to multiple replica groups with one switch, as one group
only consumes little memory. This is especially important if \sysname
is deployed in a spine switch to support many replica groups across different racks.

\subsection{Generality}
\label{sec:evaluation:generality}

We show that \sysname is a general approach by applying it to a
variety of replication protocols.
For each replication protocol, we
examine throughput for a three-replica storage system with and without
\sysname. Figure~\ref{fig:eval_generality}
shows the read throughput as
a function of write rate for different protocols

Figure~\ref{fig:eval_generality_pb} shows the results for two
primary-backup protocols, PB and CR. Both PB and CR are
limited by the performance of one server. \sysname makes use of all
three replicas to handle reads, and provides significantly higher throughput
than PB and CR. CR is able to achieve higher write throughput than PB,
as it uses a chain structure to propagate writes.

CRAQ, a modified version of CR, obtains higher read
throughput than CR, as shown in Figure~\ref{fig:eval_generality_pb}. This is
because CRAQ allows reads to be sent to any replica (reads to dirty objects are
forwarded to the tail). However, CRAQ adds an additional phase to write
operations (first marking objects as dirty then committing the write). As a
result, CRAQ's write throughput is much lower---hence the steeper curve.
\sysname(CR), which applies in-network conflict detection to CR, performs much
better than CRAQ, achieving the same level of read scalability without degrading
the performance of writes.

Figure~\ref{fig:eval_generality_quorum} shows the results for quorum-based
protocols VR and NOPaxos. For faithful comparison, we use the original
implementation of NOPaxos, including the middlebox-based sequencer
prototype, which runs on a Cavium Octeon II network processor. We
integrate \sysname with these, rather than the Tofino switch and
Redis-based backend. As a result, the absolute numbers in
Figures~\ref{fig:eval_generality_pb} and
\ref{fig:eval_generality_quorum} are incomparable.
The trends, however, are the same. \sysname significantly improves
throughput for VR and NOPaxos.

Taken together, these results demonstrate that \sysname can be applied
broadly to a wide range of replication protocols. These experiments
show the advantage of in-network conflict detection, as it introduces no
performance penalties, unlike protocol-level optimizations such as CRAQ.

\subsection{Performance Under Failures}
\label{sec:evaluation:failure}

Finally, we show how \sysname handles failures. To simulate a failure, we first
manually stop and then reactivate the switch. \sysname uses the
mechanism described in \S\ref{sec:design:failure} to correctly recover
from the failure.

Figure~\ref{fig:eval_failure} shows the throughput during this period of failure
and recovery. At time 20 s, we let the \sysname switch stop forwarding any
packets, and the system throughput drops to zero. We wait for a few seconds and
then reactivate the switch to forward packets. Upon reactivation, the switch
retains none of its former state and uses a new switch ID. The servers are
notified with the new switch ID and agree to drop single-replica reads from the
old switch. In the beginning, the switch forwards reads to the tail node and writes
to the tail node. During this time, the system throughput is the same as without
\sysname. After the first \msg{write-completion} with the new switch ID passes
the  switch, the switch has the up-to-date dirty set and last-committed point.
At this time, the switch starts scheduling single-replica reads to the servers,
and the system
throughput is fully restored. Because the servers complete requests quickly, the
transition time is minimal, and we can see that system throughput returns to
pre-failure levels within a few seconds.

\begin{figure}[t]
\centering
    \includegraphics[width=0.9\linewidth]{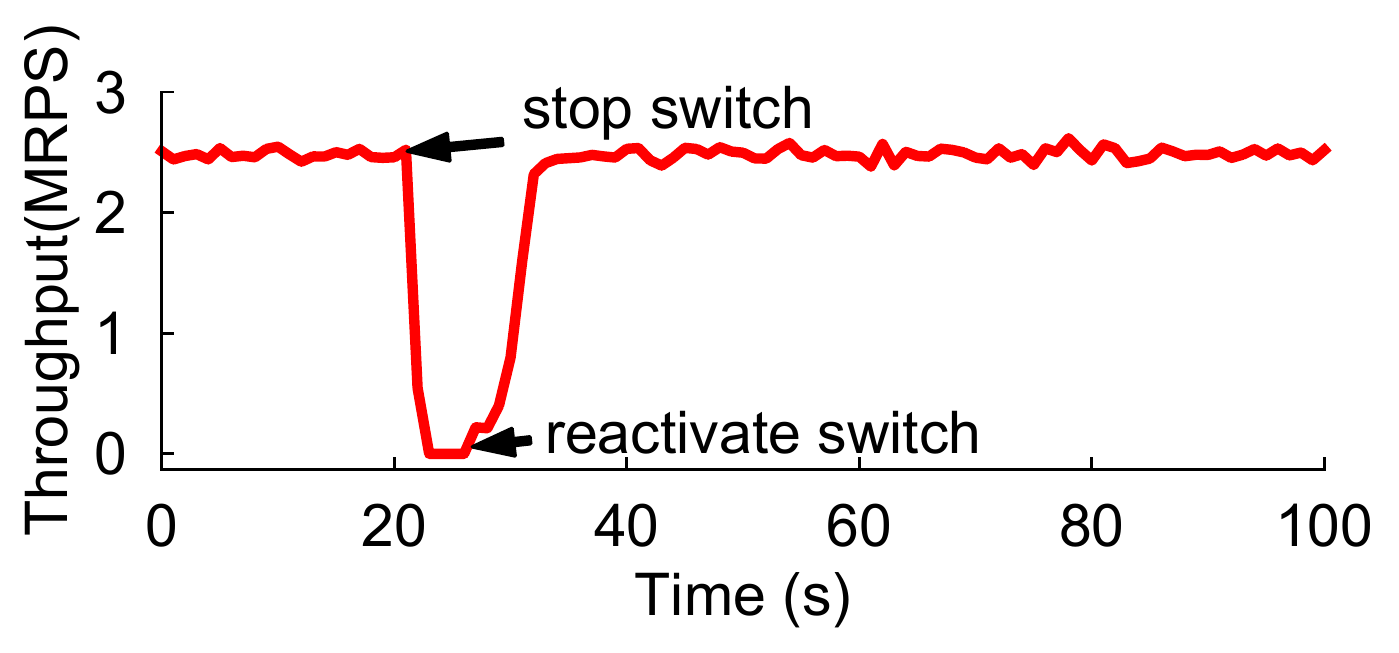}
\vspace{-0.1in}
\caption{Total throughput while the \sysname switch is stopped and then reactivated.}
\label{fig:eval_failure}
\end{figure}


\renewcommand{\secname}{Related Work}
\section{\secname}
\label{sec:related}

\para{Replication protocols.} Replication protocols are widely used by storage systems
to achieve strong consistency and fault tolerance. Dating back to
classic storage systems (e.g., Andrew~\cite{andrew}, Sprite~\cite{sprite}, Coda~\cite{coda}, Harp~\cite{harp},
RAID~\cite{raid}, Zebra~\cite{zebra}, and xFS~ \cite{xfs}), they are
now a mainstay of cloud storage services (e.g., GFS~\cite{gfs}, BigTable~\cite{bigtable},
Dynamo~\cite{dynamo}, HDFS~\cite{apache-hdfs}, Ceph~\cite{ceph},
Haystack~\cite{haystack}, f4~\cite{f4}, and Windows Azure
Storage~\cite{azure-storage}).

The primary-backup protocol~\cite{primary-backup} and its variations
like chain replication~\cite{chain-replication} and CRAQ~\cite{craq}
assign replicas with different roles (e.g., primary node, head
node, and tail node), and require operations to be executed by the replicas in a
certain order. Quorum-based protocols, such as Paxos~\cite{paxos98},
ZAB~\cite{zab}, Raft~\cite{raft}, Viewstamped
Replication~\cite{viewstamped} and Virtual Synchrony~\cite{vsync},
only require an operation to be executed at a quorum, instead of all
replicas. While they do not distinguish the roles of replicas, they
often employ an optimization that first elects a leader and then uses
the leader to commit operations to other nodes, which is very similar
to the primary-backup protocol. Vertical Paxos~\cite{vertical-paxos}
proposes to incorporate these two classes of protocols into a single
framework, by dividing a replication protocol into two parts: one is a
steady state protocol like the primary-backup protocol that optimizes
for high performance, and the other is a reconfiguration protocol like
Paxos which handles system reconfigurations, e.g., electing a leader.

CRAQ~\cite{craq} is most similar in spirit to our work. It adapts
chain replication to allow any replica to answer reads for
uncontended objects by adding a second phase to the write protocol:
objects are first marked dirty, then updated. \sysname achieves the
same goal without the write overhead by moving the contention
detection into the network, and also supports more general replication protocols.

\vspace{0.05in}
\para{Query scheduling.} A related approach is taken in a line of database
replication systems that achieve consistent transaction processing
atop multiple databases, such as
C-JDBC~\cite{cecchet04:_c_jdbc}, FAS~\cite{roehm02:_fas}, and
Ganymed~\cite{plattner04:_ganym}. These systems use a query scheduler
to orchestrate queries among replicas with different states.
The necessary logic is more complex for database transactions (and
sometimes necessitates weaker isolation levels). \sysname provides a
near-zero-overhead scheduler implementation for replication
using the network.


\vspace{0.05in}
\para{In-network computing.} The emerging programmable switches introduce new
opportunities to move computation into the network. NetCache~\cite{netcache} and
IncBricks~\cite{incbricks-asplos17} introduces in-network caching for key-value
stores. NetChain~\cite{netchain} builds a strongly-consistency, fault-tolerant,
in-network key-value store for coordination services. These designs
store object data in the switch data plane; \sysname consciously avoids
this in order to be more resource efficient. SwitchKV~\cite{switchkv}
leverages programmable switches to realize content-based routing for load
balancing in key-value stores. Eris~\cite{eris} exploits programmable switches to
realize concurrency control for distributed transactions.
NetPaxos~\cite{netpaxos, netpaxos-ccr} implements Paxos on switches.
SpecPaxos~\cite{specpaxos} and NOPaxos~\cite{nopaxos} use switches to order
messages to improve replication protocols. With NetPaxos,
SpecPaxos and NOPaxos, reads still need to be executed by a quorum, or
by a leader if the leader-based optimization is used.
\sysname improves these solutions by allowing reads not in the dirty set
to be executed by any replica.

\vspace{0.07in}


\renewcommand{\secname}{Conclusion}
\section{\secname}
\label{sec:conclusion}

In conclusion, we present \sysname, a new replicated storage architecture that achieves near-linear
scalability and guarantees linearizability with in-network conflict detection. \sysname leverages new-generation
programmable switches to efficiently track the dirty set and detect read-write conflicts
in the network data plane with no performance overhead.
Such a powerful capability enables
\sysname to safely schedule reads to the replicas without
sacrificing consistency. \sysname demonstrates that rethinking the
division of labor between the network and end hosts makes it possible
to achieve performance properties beyond the grasp of
distributed systems alone.

\para{Ethics.} This work does not raise any ethical issues.

\clearpage
{
\bibliographystyle{abbrv}
\balance
\bibliography{xin}}

\newpage
\begin{appendix}

\section{Proof of Correctness}
\label{app:proof}

\paraf{Notation.} Let $Q$ be a request, $Q.commit$ be the last-committed sequence
number added to the request by the switch, $R$ be a replica, $R.obj$ be the
local copy of an object at the replica, $R.obj.seq$ be the sequence number of
the most recent update of the object at the replica, and $R.seq$ be the sequence
number of the most recent write executed by the replica to any object.

\begin{theorem}
\sysname preserves linearizability of the replication protocols.
\label{thm:linearizability}
\end{theorem}

\begin{proof}

We prove \sysname preserves linearizability of the replication
protocols under both normal and failure scenarios. We use a fail-stop
model. All write operations are processed by the replication protocol,
and are processed in sequence number order. We need only, then,
consider the read operations. The following two properties are
sufficient for linearizability.

\begin{itemize}[leftmargin=*]
\item \textbf{P1. Visibility.} A read operation sees the effects of all
  write operations that finished before it started.
\item \textbf{P2. Integrity.} A read operation will not see the
  effects of any write operation that had not committed at the time
  the read finished.
\end{itemize}

\para{Normal scenario.} \sysname uses the dirty set to detect potential
conflicts, and if no conflicts are detected, reads are scheduled to a random
replica for better performance. \sysname leverages the last-committed sequence number to
guarantee linearizability.

For read-ahead protocols including primary-backup and chain replication, writes are committed to all replicas.
Therefore, P1 is satisfied.
However, we must also verify that a read will not see the effect of an uncommitted write.
Consider a single-replica read $Q$ that arrives at a replica $R$ to retrieve object $obj$.
It may be the case that $R$ has applied uncommitted writes to $obj$.
Therefore, $R$ compares $Q.commit$ with $R.obj.seq$.
If $Q.commit < R.obj.seq$, then $R$ forwards $Q$ for handling by the normal protocol, which will return a consistent result.
However, $Q.commit \geq R.obj.seq$, it implies that the latest write to $R.obj$ has already been committed, since writes are applied by the replication protocol in sequence number order.
Therefore, P2 is satisfied.

For read-behind protocols including Viewstamped Replication and NOPaxos, the replicas first append writes to a local log.
They execute and apply writes only after they have been committed.
As such, reading from local state will only ever reflect the results of committed writes, and P2 is satisfied.
However, we must also verify that a read will see the effect of all committed writes.
Again, consider a single-replica read $Q$ that arrives at a replica $R$ to retrieve object $obj$.
It may be the case that there were writes to $obj$ committed before $Q$ was sent which, nevertheless, $R$ has not yet executed.
Therefore, $R$ compares $Q.commit$ with $R.seq$.
If $Q.commit > R.seq$, then $R$ forwards $Q$ for handling by the normal protocol, which will return a consistent result.
However, if $Q.commit \leq R.seq$, this implies that $R$ has executed all writes to $obj$ which were committed at the time $Q$ was forwarded by the \sysname switch; otherwise the last committed sequence number on the switch would have been larger or the dirty set would have contained $obj$.
Therefore, P1 is satisfied.

\para{Switch failure.}
First, notice that in the above argument, we relied on two key facts about the state held at the switch \emph{at the time it forwards a single-replica read that will be served by a replica}.
The first is that the dirty set contains all objects with outstanding, uncommitted writes.
The second is that for all writes committed by the replication protocol, either the sequence number of that write is less than or equal to the last committed sequence number on the switch or the object being written to is in the switch's dirty set.
Now, we want to show that these two facts still hold when there are multiple \sysname switches, each receiving reads and writes.

In order for a switch to be able to forward single-replica reads at all, the switch first must have received a single \msg{write-completion} message with its switch ID.
Since sequence numbers are ordered lexicographically using the switch ID first and writes are applied in order by the replication protocol, the switch's dirty set must contain all uncommitted writes with switch ID less than or equal to its own, and all committed writes with sequence numbers with switch ID less than or equal to the switch's must either have matching entries in the dirty set or be less than or equal to the last committed sequence number.
Furthermore, if the switch forwards a single-replica read that will actually be served by a replica, that must mean that at the time it forwarded that read, no switch with larger switch ID could have yet sent any writes.
Otherwise, the replicas would have already agreed to permanently disallow single-replica reads from the switch in question.
Therefore, the switch's state, at the time it forwarded the single-replica read, was suitably up-to-date, and both of our key properties held.



\para{Server failure.} A server failure is handled by the replication protocol.
However, before a failed server is removed from the system, the protocol must ensure that the failed server is first removed from the current switch's routing information.
As long as this requirement is met, then all servers receiving single-replica reads can return linearizable results.
\end{proof}

\clearpage
\onecolumn
\section{Harmonia Specification}
\label{app:tla}
\definecolor{boxshade}{gray}{0.85}
\setlength{\textwidth}{360pt}
\setlength{\textheight}{541pt}

\tlatex
\setboolean{shading}{true}
\@x{}\moduleLeftDash\@xx{ {\MODULE} harmonia}\moduleRightDash\@xx{}%
\begin{lcom}{5.0}%
\begin{cpar}{0}{F}{F}{0}{0}{}%
Specifies the \ensuremath{Harmonia} protocol.
\end{cpar}%
\end{lcom}%
\@x{ {\EXTENDS} Naturals ,\, FiniteSets ,\, Sequences ,\, TLC}%
\@pvspace{8.0pt}%
\@x{}\midbar\@xx{}%
\@x{}%
\@y{\@s{0}%
   \textbf{\large Constants and Variables}
}%
\@xx{}%
\@pvspace{8.0pt}%
\@x{ {\CONSTANTS} dataItems ,\,\@s{13.52}}%
\@y{\@s{0}%
 Set of model values representing data items
}%
\@xx{}%
\@x{\@s{54.75} numSwitches ,\,}%
\@y{\@s{0}%
 Number of total switches
}%
\@xx{}%
\@x{\@s{54.75} replicas ,\,\@s{23.89}}%
\@y{\@s{0}%
 Set of model values representing the replicas
}%
\@xx{}%
\@x{\@s{54.75} isReadBehind\@s{3.32}}%
\@y{\@s{0}%
 Whether the replication protocol is read-behind
}%
\@xx{}%
\@pvspace{8.0pt}%
\@x{ {\ASSUME} \.{\land} IsFiniteSet ( dataItems )}%
\@x{\@s{38.24} \.{\land} IsFiniteSet ( replicas )}%
\@x{\@s{38.24} \.{\land} numSwitches \.{>} 0}%
\@x{\@s{38.24} \.{\land} isReadBehind \.{\in} \{ {\TRUE} ,\, {\FALSE} \}}%
\@pvspace{8.0pt}%
\@x{ isReadAhead \.{\defeq} {\lnot} isReadBehind}%
\@pvspace{8.0pt}%
\@x{ {\VARIABLE} messages ,\,\@s{45.77}}%
\@y{\@s{0}%
 The network, a set of all messages sent
}%
\@xx{}%
\@x{\@s{46.84} switchStates ,\,\@s{32.23}}%
\@y{\@s{0}%
 The state of the \ensuremath{Harmonia} switches
}%
\@xx{}%
\@x{\@s{46.84} activeSwitch ,\,\@s{31.52}}%
\@y{\@s{0}%
 The switch allowed to send \ensuremath{Harmonia} reads
}%
\@xx{}%
\@x{\@s{46.84} sharedLog ,\,\@s{42.89}}%
\@y{\@s{0}%
 The main \ensuremath{log} decided on by replication protocol
}%
\@xx{}%
\@x{\@s{46.84} replicaCommitPoints}%
\@y{\@s{0}%
 The latest write processed by each replica
}%
\@xx{}%
\@pvspace{8.0pt}%
\@x{}%
\@y{\@s{0}%
 A value smaller than all writes sent by switches
}%
\@xx{}%
 \@x{ BottomWrite \.{\defeq} [ switchNum \.{\mapsto} 0 ,\, seq \.{\mapsto} 0
 ]}%
\@pvspace{8.0pt}%
\begin{lcom}{5.0}%
\begin{cpar}{0}{F}{F}{0}{0}{}%
  \textbf{Message Schemas}
\end{cpar}%
\vshade{5.0}%
\begin{cpar}{0}{F}{F}{0}{0}{}%
Write (Switch to Replication Protocol)
\end{cpar}%
\begin{cpar}{0}{T}{F}{10.0}{0}{}%
[ \ensuremath{mtype\@s{12.5}\.{\mapsto} MWrite},
\end{cpar}%
\begin{cpar}{0}{T}{F}{5.0}{0}{}%
 \ensuremath{switchNum \.{\mapsto} i \.{\in}}
 (\ensuremath{1\.{\dotdot}numSwitches}),
\end{cpar}%
\begin{cpar}{0}{F}{F}{0}{0}{}%
\ensuremath{seq\@s{17.5}\.{\mapsto} i \.{\in}} (\ensuremath{1\.{\dotdot}}),
\end{cpar}%
\begin{cpar}{0}{F}{F}{0}{0}{}%
\ensuremath{dataItem \.{\mapsto} d \.{\in} dataItems} ]
\end{cpar}%
\vshade{5.0}%
\begin{cpar}{2}{F}{F}{0}{0}{}%
 The \ensuremath{ProtocolRead}, \ensuremath{HarmoniaRead}, and
 \ensuremath{ReadResponse} messages contain a field
 (\ensuremath{ghostLastReponse}) which is not used in the protocol, and is
 only
 present to aid in the definition of linearizability.
\end{cpar}%
\vshade{5.0}%
\begin{cpar}{0}{F}{F}{0}{0}{}%
\ensuremath{ProtocolRead} (Switch to Replication Protocol)
\end{cpar}%
\begin{cpar}{0}{T}{F}{10.0}{0}{}%
[ \ensuremath{mtype\@s{30.0}\.{\mapsto} MProtocolRead},
\end{cpar}%
\begin{cpar}{0}{T}{F}{5.0}{0}{}%
\ensuremath{dataItem\@s{22.5}\.{\mapsto} d \.{\in} dataItems},
\end{cpar}%
\begin{cpar}{0}{F}{F}{0}{0}{}%
 \ensuremath{ghostLastReponse \.{\mapsto} w \.{\in} WRITES} (the set of all
 \ensuremath{MWrite} messages) ]
\end{cpar}%
\vshade{5.0}%
\begin{cpar}{2}{F}{F}{0}{0}{}%
\ensuremath{HarmoniaRead} (Switch to Replicas)
\end{cpar}%
\begin{cpar}{0}{T}{F}{10.0}{0}{}%
[ \ensuremath{mtype\@s{30.0}\.{\mapsto} MHarmoniaRead},
\end{cpar}%
\begin{cpar}{0}{T}{F}{5.0}{0}{}%
\ensuremath{dataItem\@s{22.5}\.{\mapsto} d \.{\in} dataItems},
\end{cpar}%
\begin{cpar}{0}{F}{F}{0}{0}{}%
 \ensuremath{switchNum\@s{20.0}\.{\mapsto} i \.{\in}}
 (\ensuremath{1\.{\dotdot}numSwitches}),
\end{cpar}%
\begin{cpar}{0}{F}{F}{0}{0}{}%
\ensuremath{lastCommitted \.{\mapsto} w \.{\in} WRITES},
\end{cpar}%
\begin{cpar}{0}{F}{F}{0}{0}{}%
\ensuremath{ghostLastReponse \.{\mapsto} w \.{\in} WRITES} ]
\end{cpar}%
\vshade{5.0}%
\begin{cpar}{2}{F}{F}{0}{0}{}%
\ensuremath{ReadResponse} (Replicas/Replication Protocol to Client)
\end{cpar}%
\begin{cpar}{0}{T}{F}{10.0}{0}{}%
[ \ensuremath{mtype\@s{30.0}\.{\mapsto} MReadResponse},
\end{cpar}%
\begin{cpar}{0}{T}{F}{5.0}{0}{}%
\ensuremath{write\@s{30.0}\.{\mapsto} w \.{\in} WRITES},
\end{cpar}%
\begin{cpar}{0}{F}{F}{0}{0}{}%
\ensuremath{ghostLastReponse \.{\mapsto} w \.{\in} WRITES}, ]
\end{cpar}%
\end{lcom}%
\@x{ {\CONSTANTS} MWrite ,\,}%
\@x{\@s{54.75} MProtocolRead ,\,}%
\@x{\@s{54.75} MHarmoniaRead ,\,}%
\@x{\@s{54.75} MReadResponse}%
\@pvspace{8.0pt}%
\@x{ Init \.{\defeq} \.{\land} messages \.{=} \{ \}}%
 \@x{\@s{35.70} \.{\land} switchStates\@s{0.71} \.{=} [ i \.{\in} ( 1
 \.{\dotdot} numSwitches ) \.{\mapsto}}%
\@x{\@s{118.68} [ seq\@s{49.98} \.{\mapsto} 0 ,\,}%
 \@x{\@s{121.45} dirtySet\@s{28.57} \.{\mapsto} [ d \.{\in} \{ \} \.{\mapsto}
 0 ] ,\,}%
\@x{\@s{121.45} lastCommitted \.{\mapsto} BottomWrite ] ]}%
\@x{\@s{35.70} \.{\land} activeSwitch \.{=} 1}%
 \@x{\@s{35.70} \.{\land} replicaCommitPoints \.{=} [ r \.{\in} replicas
 \.{\mapsto} 0 ]}%
\@x{\@s{35.70} \.{\land} sharedLog \.{=} {\langle} {\rangle}}%
\@pvspace{8.0pt}%
\@x{}\midbar\@xx{}%
\@x{}%
\@y{\@s{0}%
   \textbf{\large Helper and Utility Functions}
}%
\@xx{}%
\@pvspace{8.0pt}%
\@x{}%
\@y{\@s{0}%
 Basic utility functions
}%
\@xx{}%
\@x{ Range ( t ) \.{\defeq} \{ t [ i ] \.{:} i \.{\in} {\DOMAIN} t \}}%
\@pvspace{8.0pt}%
 \@x{ Min ( S )\@s{1.49} \.{\defeq} {\CHOOSE} s \.{\in} S \.{:} \A\, sp
 \.{\in} S \.{:} sp \.{\geq} s}%
\@pvspace{8.0pt}%
 \@x{ Max ( S ) \.{\defeq} {\CHOOSE} s \.{\in} S \.{:} \A\, sp \.{\in} S \.{:}
 sp \.{\leq} s}%
\@pvspace{8.0pt}%
\@x{}%
\@y{\@s{0}%
 Sequence number functions
}%
\@xx{}%
 \@x{ GTE ( w1 ,\, w2 ) \.{\defeq} \.{\lor} w1 . switchNum \.{>} w2 .
 switchNum}%
\@x{\@s{81.09} \.{\lor} \.{\land} w1 . switchNum \.{=} w2 . switchNum}%
\@x{\@s{92.20} \.{\land} w1 . seq \.{\geq} w2 . seq}%
\@pvspace{8.0pt}%
 \@x{ GT ( w1 ,\, w2 ) \.{\defeq} \.{\lor} w1 . switchNum \.{>} w2 .
 switchNum}%
\@x{\@s{78.54} \.{\lor} \.{\land} w1 . switchNum \.{=} w2 . switchNum}%
\@x{\@s{89.65} \.{\land} w1 . seq \.{>} w2 . seq}%
\@pvspace{8.0pt}%
 \@x{ MinW ( W )\@s{1.06} \.{\defeq} {\CHOOSE} w \.{\in} W \.{:} \A\, wp
 \.{\in} W \.{:} GTE ( wp ,\, w )}%
\@pvspace{8.0pt}%
 \@x{ MaxW ( W ) \.{\defeq} {\CHOOSE} w \.{\in} W \.{:} \A\, wp \.{\in} W
 \.{:} GTE ( w ,\, wp )}%
\@pvspace{8.0pt}%
\@x{}%
\@y{\@s{0}%
 Common \ensuremath{log}-processing functions
}%
\@xx{}%
\@x{ CommittedLog \.{\defeq} {\IF} isReadBehind}%
\@x{\@s{82.26} \.{\THEN} sharedLog}%
 \@x{\@s{82.26} \.{\ELSE} SubSeq ( sharedLog ,\, 1 ,\, Min ( Range (
 replicaCommitPoints ) ) )}%
\@pvspace{8.0pt}%
\@x{ MaxCommittedWriteForIn ( d ,\, log ) \.{\defeq}}%
 \@x{\@s{8.2} MaxW ( \{ BottomWrite \} \.{\cup} \{ m \.{\in} Range ( log )
 \.{:} m . dataItem \.{=} d \} )}%
\@pvspace{8.0pt}%
 \@x{ MaxCommittedWriteFor ( d ) \.{\defeq} MaxCommittedWriteForIn ( d ,\,
 CommittedLog )}%
\@pvspace{8.0pt}%
 \@x{ MaxCommittedWrite \.{\defeq} MaxW ( \{ BottomWrite \} \.{\cup} Range (
 CommittedLog ) )}%
\@pvspace{8.0pt}%
\@x{}%
\@y{\@s{0}%
 Short-hand way of sending a message
}%
\@xx{}%
\@x{ Send ( m ) \.{\defeq} messages \.{'} \.{=} messages \.{\cup} \{ m \}}%
\@pvspace{8.0pt}%
\@x{}\midbar\@xx{}%
\@x{}%
\@y{\@s{0}%
   \textbf{\large Main Spec}
}%
\@xx{}%
\@pvspace{8.0pt}%
\begin{lcom}{5.0}%
\begin{cpar}{0}{F}{F}{0}{0}{}%
Using the ghost variables in reads and read responses, we can define our main
 safety property (linearizability) rather simply.
\end{cpar}%
\end{lcom}%
 \@x{ Linearizability \.{\defeq} \A\, m \.{\in} \{ mp \.{\in} messages \.{:}
 mp . mtype \.{=} MReadResponse \} \.{:}}%
\@x{\@s{81.88} \.{\land} GTE ( m . write ,\, m . ghostLastReponse )}%
\@x{\@s{81.88} \.{\land} \.{\lor} m . write \.{\in} Range ( CommittedLog )}%
\@x{\@s{92.99} \.{\lor} m . write \.{=} BottomWrite}%
\@pvspace{8.0pt}%
\@x{}\midbar\@xx{}%
\@x{}%
\@y{\@s{0}%
   \textbf{\large Actions and Message Handlers}
}%
\@xx{}%
\@pvspace{8.0pt}%
\@x{}%
\@y{\@s{0}%
 Switch \ensuremath{s} sends write for data item \ensuremath{d
}}%
\@xx{}%
\@x{ SendWrite ( s ,\, d ) \.{\defeq}}%
\@x{ \.{\LET}}%
\@x{\@s{8.2} nextSeq \.{\defeq} switchStates [ s ] . seq \.{+} 1}%
\@x{ \.{\IN}}%
\@x{\@s{8.2}}%
\@y{\@s{0}%
 Only activated switches can send writes
}%
\@xx{}%
\@x{\@s{8.2} \.{\land} s \.{\leq} activeSwitch}%
 \@x{\@s{8.2} \.{\land} switchStates \.{'} \.{=} [ switchStates {\EXCEPT}
 {\bang} [ s ] \.{=}}%
\@x{\@s{19.31} [ @ {\EXCEPT} {\bang} . seq \.{=} nextSeq ,\,}%
 \@x{\@s{71.88} {\bang} . dirtySet \.{=} ( d \.{\colongt} nextSeq ) \.{\,@@\,}
 @ ] ]}%
\@x{\@s{8.2} \.{\land} Send ( [ mtype\@s{22.27} \.{\mapsto} MWrite ,\,}%
\@x{\@s{47.96} switchNum \.{\mapsto} s ,\,}%
\@x{\@s{47.96} seq\@s{34.91} \.{\mapsto} nextSeq ,\,}%
\@x{\@s{47.96} dataItem\@s{9.19} \.{\mapsto} d ] )}%
 \@x{\@s{8.2} \.{\land} {\UNCHANGED} {\langle} activeSwitch ,\,
 replicaCommitPoints ,\, sharedLog {\rangle}}%
\@pvspace{8.0pt}%
\@x{}%
\@y{\@s{0}%
 Add write \ensuremath{w} to the shared \ensuremath{log
}}%
\@xx{}%
\@x{ HandleWrite ( w ) \.{\defeq}}%
\@x{\@s{8.2}}%
\@y{\@s{0}%
 The replication protocol adds writes in order
}%
\@xx{}%
\@x{\@s{8.2} \.{\land} \.{\lor} Len ( sharedLog ) \.{=} 0}%
\@x{\@s{19.31} \.{\lor} \.{\land} Len ( sharedLog ) \.{>} 0}%
\@x{\@s{30.42} \.{\land} GTE ( w ,\, sharedLog [ Len ( sharedLog ) ] )}%
\@x{\@s{8.2} \.{\land} sharedLog \.{'} \.{=} Append ( sharedLog ,\, w )}%
 \@x{\@s{8.2} \.{\land} {\UNCHANGED} {\langle} messages ,\, switchStates ,\,
 activeSwitch ,\, replicaCommitPoints {\rangle}}%
\@pvspace{8.0pt}%
\@x{}%
\@y{\@s{0}%
 The switch processes a write completion for write \ensuremath{w
}}%
\@xx{}%
\@x{ ProcessWriteCompletion ( w ) \.{\defeq}}%
\@x{ \.{\LET}}%
\@x{\@s{8.2} s \.{\defeq} w . switchNum}%
\@x{\@s{8.2} ds \.{\defeq} switchStates [ s ] . dirtySet}%
 \@x{\@s{8.2} dsp \.{\defeq} [ dp \.{\in} \{ d \.{\in} {\DOMAIN} ds \.{:} ds [
 d ] \.{>} w . seq \} \.{\mapsto} ds [ dp ] ]}%
\@x{ \.{\IN}}%
\@x{\@s{8.2}}%
\@y{\@s{0}%
 Write is committed (processed by all in read-ahead mode)
}%
\@xx{}%
\@x{\@s{8.2} \.{\land} GTE ( MaxCommittedWrite ,\, w )}%
 \@x{\@s{8.2} \.{\land} switchStates \.{'} \.{=} [ switchStates {\EXCEPT}
 {\bang} [ s ] \.{=}}%
\@x{\@s{19.31} [ @ {\EXCEPT} {\bang} . dirtySet \.{=} dsp ,\,}%
\@x{\@s{71.88} {\bang} . lastCommitted \.{=} MaxW ( \{ @ ,\, w \} ) ] ]}%
 \@x{\@s{8.2} \.{\land} {\UNCHANGED} {\langle} messages ,\, activeSwitch ,\,
 replicaCommitPoints ,\, sharedLog {\rangle}}%
\@pvspace{8.0pt}%
\@x{}%
\@y{\@s{0}%
 Replica \ensuremath{r} locally commits the next write from the shared
 \ensuremath{log
}}%
\@xx{}%
\@x{ CommitWrite ( r ) \.{\defeq}}%
\@x{\@s{8.2} \.{\land} Len ( sharedLog ) \.{>} replicaCommitPoints [ r ]}%
 \@x{\@s{8.2} \.{\land} replicaCommitPoints \.{'} \.{=} [ replicaCommitPoints
 {\EXCEPT} {\bang} [ r ] \.{=} @ \.{+} 1 ]}%
 \@x{\@s{8.2} \.{\land} {\UNCHANGED} {\langle} messages ,\, switchStates ,\,
 activeSwitch ,\, sharedLog {\rangle}}%
\@pvspace{8.0pt}%
\@x{}%
\@y{\@s{0}%
 Switch \ensuremath{s} sends read for data item \ensuremath{d
}}%
\@xx{}%
\@x{ SendRead ( s ,\, d ) \.{\defeq}}%
\@x{ \.{\LET}}%
 \@x{\@s{8.2} shouldSendHarmoniaRead \.{\defeq} \.{\land} d \.{\notin}
 {\DOMAIN} switchStates [ s ] . dirtySet}%
\@x{\@s{141.36}}%
\@y{\@s{0}%
 Can send \ensuremath{Harmonia} reads after one completion
}%
\@xx{}%
 \@x{\@s{141.36} \.{\land} GT ( switchStates [ s ] . lastCommitted ,\,
 BottomWrite )}%
\@x{\@s{8.2} returnedReads \.{\defeq} \{ m . write \.{:}}%
 \@x{\@s{16.4} m \.{\in} \{ mp \.{\in} messages \.{:} \.{\land} mp . mtype
 \.{=} MReadResponse}%
\@x{\@s{119.44} \.{\land} mp . write\@s{3.70} \.{\neq} BottomWrite}%
\@x{\@s{119.44} \.{\land} mp . write . dataItem \.{=} d \} \}}%
 \@x{\@s{8.2} lr \.{\defeq} MaxW ( \{ MaxCommittedWriteFor ( d ) \} \.{\cup}
 returnedReads )}%
\@x{ \.{\IN}}%
\@x{\@s{8.2} \.{\land} \.{\lor} \.{\land} shouldSendHarmoniaRead}%
 \@x{\@s{30.42} \.{\land} Send ( [ mtype\@s{50.36} \.{\mapsto} MHarmoniaRead
 ,\,}%
\@x{\@s{70.18} dataItem\@s{37.29} \.{\mapsto} d ,\,}%
\@x{\@s{70.18} switchNum\@s{28.09} \.{\mapsto} s ,\,}%
 \@x{\@s{70.18} lastCommitted\@s{13.02} \.{\mapsto} switchStates [ s ] .
 lastCommitted ,\,}%
\@x{\@s{70.18} ghostLastReponse \.{\mapsto} lr ] )}%
\@x{\@s{19.31} \.{\lor} \.{\land} {\lnot} shouldSendHarmoniaRead}%
 \@x{\@s{30.42} \.{\land} Send ( [ mtype\@s{50.36} \.{\mapsto} MProtocolRead
 ,\,}%
\@x{\@s{70.18} dataItem\@s{37.29} \.{\mapsto} d ,\,}%
\@x{\@s{70.18} ghostLastReponse \.{\mapsto} lr ] )}%
 \@x{\@s{8.2} \.{\land} {\UNCHANGED} {\langle} switchStates ,\, activeSwitch
 ,\, replicaCommitPoints ,\, sharedLog {\rangle}}%
\@pvspace{8.0pt}%
\@x{}%
\@y{\@s{0}%
 Process protocol read \ensuremath{m
}}%
\@xx{}%
\@x{ HandleProtocolRead ( m ) \.{\defeq}}%
 \@x{\@s{8.2} \.{\land} Send ( [ mtype\@s{50.36} \.{\mapsto} MReadResponse
 ,\,}%
 \@x{\@s{47.96} write\@s{54.07} \.{\mapsto} MaxCommittedWriteFor ( m .
 dataItem ) ,\,}%
\@x{\@s{47.96} ghostLastReponse \.{\mapsto} m . ghostLastReponse ] )}%
 \@x{\@s{8.2} \.{\land} {\UNCHANGED} {\langle} switchStates ,\, activeSwitch
 ,\, replicaCommitPoints ,\, sharedLog {\rangle}}%
\@pvspace{8.0pt}%
\@x{}%
\@y{\@s{0}%
 Replica \ensuremath{r} receives \ensuremath{Harmonia} read \ensuremath{r
}}%
\@xx{}%
\@x{ HandleHarmoniaRead ( r ,\, m ) \.{\defeq}}%
\@x{ \.{\LET}}%
\@x{\@s{8.2} cp \.{\defeq} replicaCommitPoints [ r ]}%
 \@x{\@s{8.2} w \.{\defeq} MaxCommittedWriteForIn ( m . dataItem ,\, SubSeq (
 sharedLog ,\, 1 ,\, cp ) )}%
\@x{ \.{\IN}}%
\@x{\@s{8.2}}%
\@y{\@s{0}%
 Can only accept \ensuremath{Harmonia} reads from the active switch
}%
\@xx{}%
\@x{\@s{8.2} \.{\land} m . switchNum \.{=} activeSwitch}%
\@x{\@s{8.2} \.{\land} \.{\lor} \.{\land} isReadBehind}%
\@x{\@s{30.42}}%
\@y{\@s{0}%
 Replica can only process read if it is up-to-date
}%
\@xx{}%
 \@x{\@s{30.42} \.{\land} GTE ( {\IF} cp \.{>} 0 \.{\THEN} sharedLog [ cp ]
 \.{\ELSE} BottomWrite ,\, m . lastCommitted )}%
\@x{\@s{19.31} \.{\lor} \.{\land} isReadAhead}%
\@x{\@s{30.42}}%
\@y{\@s{0}%
 Replica can only process read if write was completed
}%
\@xx{}%
\@x{\@s{30.42} \.{\land} GTE ( m . lastCommitted ,\, w )}%
 \@x{\@s{8.2} \.{\land} Send ( [ mtype\@s{50.36} \.{\mapsto} MReadResponse
 ,\,}%
\@x{\@s{47.96} write\@s{54.07} \.{\mapsto} w ,\,}%
\@x{\@s{47.96} ghostLastReponse \.{\mapsto} m . ghostLastReponse ] )}%
 \@x{\@s{8.2} \.{\land} {\UNCHANGED} {\langle} switchStates ,\, activeSwitch
 ,\, replicaCommitPoints ,\, sharedLog {\rangle}}%
\@pvspace{8.0pt}%
\@x{ SwitchFailover \.{\defeq}}%
\@x{\@s{8.2} \.{\land} activeSwitch \.{<} numSwitches}%
\@x{\@s{8.2} \.{\land} activeSwitch \.{'} \.{=} activeSwitch \.{+} 1}%
 \@x{\@s{8.2} \.{\land} {\UNCHANGED} {\langle} messages ,\, switchStates ,\,
 replicaCommitPoints ,\, sharedLog {\rangle}}%
\@pvspace{8.0pt}%
\@x{}\midbar\@xx{}%
\@x{}%
\@y{\@s{0}%
   \textbf{\large Main Transition Function}
}%
\@xx{}%
\@pvspace{8.0pt}%
 \@x{ Next \.{\defeq} \.{\lor} \E\, s\@s{4.03} \.{\in} ( 1 \.{\dotdot}
 numSwitches ) \.{:}}%
 \@x{\@s{50.94} \E\, d\@s{2.80} \.{\in} dataItems \.{:} \.{\lor} SendWrite ( s
 ,\, d )}%
\@x{\@s{132.29} \.{\lor} SendRead ( s ,\, d )}%
 \@x{\@s{39.83} \.{\lor} \E\, m \.{\in} Range ( sharedLog ) \.{:}
 ProcessWriteCompletion ( m )}%
 \@x{\@s{39.83} \.{\lor} \E\, m \.{\in} messages \.{:} \.{\lor} \.{\land} m .
 mtype \.{=} MWrite}%
\@x{\@s{140.06} \.{\land} HandleWrite ( m )}%
\@x{\@s{128.95} \.{\lor} \.{\land} m . mtype \.{=} MProtocolRead}%
\@x{\@s{140.06} \.{\land} HandleProtocolRead ( m )}%
\@x{\@s{128.95} \.{\lor} \E\, r \.{\in} replicas \.{:}}%
\@x{\@s{140.06} \.{\land} m . mtype \.{=} MHarmoniaRead}%
\@x{\@s{140.06} \.{\land} HandleHarmoniaRead ( r ,\, m )}%
 \@x{\@s{39.83} \.{\lor} \E\, r\@s{3.65} \.{\in} replicas\@s{7.02} \.{:}
 CommitWrite ( r )}%
\@x{\@s{39.83} \.{\lor} SwitchFailover}%
\@pvspace{8.0pt}%
\@x{}\bottombar\@xx{}%

\end{appendix}

\end{document}